\documentclass[11pt]{article}
\usepackage[T1]{fontenc}
\usepackage[latin9]{inputenc}
\usepackage{amsmath}
\usepackage{amssymb}
\usepackage{mathtools}
\usepackage{mathtools}
\usepackage{amsthm}
\usepackage[linesnumbered,ruled,vlined]{algorithm2e}
\usepackage{graphicx}
\usepackage{pgfplots}
\pgfplotsset{compat=1.18}
\usepackage{subcaption}
\usepackage{amsmath}
\usepackage{appendix}
\usepackage[colorlinks=true,citecolor=blue,linkcolor=blue,urlcolor=black,pdfborder={0 0 0}]{hyperref}
\usepackage{cleveref}
\usepackage{xcolor}
\usepackage{comment}
\usepackage{thm-restate}
\usepackage{natbib}
\usepackage{fullpage}
\usepackage{array} 
\usepackage{tikz}
\usepackage{caption}
\usepackage{enumitem}
\usepackage{booktabs} 
\usepackage{bm}
\usepackage{array}
\newcolumntype{P}[1]{>{\centering\arraybackslash}p{#1}}

\allowdisplaybreaks

\newcommand\blfootnote[1]{%
  \begingroup
  \renewcommand\thefootnote{}%
  \NoHyper\footnote{#1}\endNoHyper
  \addtocounter{footnote}{-1}%
  \endgroup
}

\newtheorem{theorem}{Theorem}[section]
\newtheorem{theoremnatural}{Theorem}
\newtheorem{proposition}[theorem]{Proposition}
\newtheorem{lemma}[theorem]{Lemma}
\newtheorem{claim}[theorem]{Claim}
\newtheorem{corollary}[theorem]{Corollary}
\newtheorem{definition}[theorem]{Definition}

\newtheorem{observation}[theorem]{Observation}

\newtheorem*{theorem*}{Theorem}

\newcommand{\reals}{\mathbb{R}}
\newcommand{\indicator}{\textbf{1}}

\DeclarePairedDelimiter{\ceil}{\lceil}{\rceil}
\newcommand{\eps}{\varepsilon}
\DeclareMathOperator*{\argmax}{arg\,max}

\newcommand{\cmmnt}[1]{}

\newcommand{\agents}{A}

\newcommand{\goodobj}{BEST}

\newcommand{\MAX}{\textsc{Max-}}

\newcommand{\maxrevenue}{\textsc{Max-Profit}}

\newcommand{\nash}{\mathsf{NE}}

\newcommand{\noindexcontract}{\boldsymbol{\alpha}}
\newcommand{\icontract}{\alpha_i}
\newcommand{\someindexcontract}{\alpha}
\newcommand{\setofcontracts}{\mathcal{C}(B)}
\newcommand{\noindexsubcontract}[1]{\noindexcontract|_{#1}}
\newcommand{\noindexsubcontractstar}[1]{\noindexcontract^\star|_{#1}}
\newcommand{\epsilonvec}{\boldsymbol{\eps}}

\newcommand{\multiInstance}{\langle \agents, \{T_i\}_{i \in \agents}, f, c \rangle}

\newcommand{\multimlr}{\textsc{Max-Reward-Bounded}}
\newcommand{\bestsingle}[2]{\textsc{Best-Single}_#1\textsc{-}#2}
\newcommand{\specialAgents}{A'}

\newcommand{\goodaction}{\mathcal{G}}
\newcommand{\badaction}{\mathcal{B}}
\newcommand{\actions}{T}

\title{One Action Too Many:\\ Inapproximability of Budgeted Combinatorial Contracts}

\author{ 
Michal Feldman$^\ast$
\quad
Yoav Gal-Tzur$^\dagger$
\quad 
Tomasz Ponitka$^\ddagger$
\quad
Maya Schlesinger$^\mathsection$
}
\date{\today}

\begin{document}

\maketitle

\blfootnote{
This is the full version of an ITCS 2026 paper.
This project has been partially funded by the European Research Council (ERC) under the European Union's Horizon 2020 program (grant agreement No.~866132), by the European Union's Horizon Europe Program (grant agreement No.~101170373), by an Amazon Research Award, by the Israel Science Foundation Breakthrough Program (grant No.~2600/24), and by a grant from TAU Center for AI and Data Science (TAD), and by the NSF-BSF (grant number 2020788).}
\blfootnote{$^\ast$Tel Aviv University and Microsoft ILDC, Israel. Email: \texttt{mfeldman@tauex.tau.ac.il}}
\blfootnote{$^\dagger$Tel Aviv University, Israel. Email: \texttt{yoavgaltzur@mail.tau.ac.il}}
\blfootnote{$^\ddagger$Tel Aviv University, Israel. Email: \texttt{tomaszp@mail.tau.ac.il}}
\blfootnote{$^\mathsection$Tel Aviv University, Israel. Email: \texttt{mayas1@mail.tau.ac.il}}

\begin{abstract}
We study multi-agent contract design with combinatorial actions, under budget constraints, and for a broad class of objective functions, including profit (principal's utility), reward, and welfare. 
Our first result is a strong impossibility: For submodular reward functions, no randomized poly-time algorithm can approximate the optimal budget-feasible value within \textit{any finite factor}, even with demand-oracle access.
This result rules out extending known constant-factor guarantees from either (i) unbudgeted settings with combinatorial actions or (ii) budgeted settings with binary actions, to their combination. The hardness is tight: It holds even when all but one agent have binary actions and the remaining agent has just one additional action. 
On the positive side, we show that gross substitutes rewards (a well-studied strict subclass of submodular functions) admit a deterministic poly-time $O(1)$-approximation, using only value queries. Our results thus draw the first sharp separation between budgeted and unbudgeted settings in combinatorial contracts, and identifies gross substitutes as a tractable frontier for budgeted combinatorial contracts. 
Finally, we present an FPTAS for additive rewards, demonstrating that arbitrary approximation is tractable under any budget. This constitutes the first FPTAS for the multi-agent combinatorial-actions setting, even in the absence of budget constraints.
\end{abstract}
\setcounter{page}{0}
\thispagestyle{empty}

\newpage
\section{Introduction}

Contract design is a central area in microeconomics that studies how to incentivize strategic agents to exert costly effort in their tasks \cite[e.g.,][]{holmstrom1979moral,grossman1992analysis,ross1973economic,innes1990limited,holmstrom1982moral}. The significance of this field to modern economics was highlighted by the 2016 Nobel Prize awarded to Hart and H\"olmstrom for their contributions to contract theory \cite{Nobel2016}. In recent years, the rapid emergence of computerized markets for services has fueled increasing interest in the algorithmic aspects of contract design \cite[e.g.,][]{dutting2019simple,dutting2022combinatorial,babaioff2006combinatorial,castiglioni2023multi}; see \cite{DuttingFT24,Feldman25} for recent surveys.

In this work, we focus on the algorithmic challenges of incentivizing teamwork \cite[e.g.,][]{holmstrom1982moral,babaioff2006combinatorial,duetting2022multi}. Consider the following illustrative example of a startup hiring a team of engineers. The startup seeks to build a successful product, whose probability of success depends both on the team's composition (for instance, the degree to which the engineers' expertise overlaps or complements each other), and on each individual's effort (for example, whether every engineer carefully reviews code). To encourage engineers to exert costly effort, the startup offers them equity: each employee receives a fraction of ownership in the company.

Scenarios of this type are captured by the multi-agent combinatorial-actions model for contract design~\cite{duetting2022multi,multimulti}.
In this model a principal (she, the startup founder) delegates a binary-outcome project to a pool of $n$ potential agents (engineers). The project yields the principal a reward of $1$ upon success and $0$ if it fails. 
Each agent $i \in [n]$ has a set $T_i$ of available actions, and may choose any subset $S_i \subseteq T_i$ of them.
Each action $j \in T_i$ has a non-negative cost $c_j$, and the agent $i$ pays the total cost of their chosen actions, $\sum_{j \in S_i} c_j$.
A set function $f$ maps every profile of actions $S = S_1\sqcup \ldots \sqcup S_n$ to the project's success probability $f(S)$. 
Since the reward is normalized to $1$, $f(S)$ also denotes the \textit{expected reward}. To incentivize effort, the principal offers each agent a \emph{linear contract} $\alpha_i \in [0,1]$, 
denoting the fraction of the reward that goes to agent $i$ if the project succeeds.
So agent $i$'s utility, under action profile $S$, is $\alpha_i f(S) - \sum_{j \in S_i} c_j$.
Every contract $\noindexcontract=(\alpha_1, \ldots, \alpha_n)$ induces a game among the agents. An action profile $S$ is a {\em Nash equilibrium} of the game if, given the contract offered, no agent can increase his utility by deviating.
In this case, we say that the contract incentivizes action profile $S$.

The principal's objective in \cite{multimulti} is to incentivize an action profile $S$ maximizing her \emph{profit}, defined as her expected reward minus the total expected payments to the agents, i.e., $f(S) \cdot (1 - \sum_{i \in [n]} \alpha_i)$. 
They present an efficient algorithm that, for submodular $f$, attains a constant-factor approximation to the optimal profit using demand queries to $f$ (see \Cref{sec:model} for definitions).

However, practical applications may introduce additional objectives and constraints.
Principals may pursue objectives that go beyond profit maximization; for example, maximizing the probability of success $f(S)$, or maximizing total welfare, $f(S) - \sum_{i \in [n]} \sum_{j \in S_i} c_j$.
At the same time, they often face constraints, such as budget constraints (e.g., a given portion of the company's equity must be allocated to investors).

Polynomial-time approximation algorithms for multi-agent contract settings, under key objectives, including reward and welfare, and budget constraints, have been obtained by \cite{feldman2025budget,aharoni2025welfare}, but only within the \emph{binary-actions} model, where each agent either exerts effort or shirks.
In this work, we extend the multi-agent \emph{combinatorial-actions} model of \cite{multimulti} to incorporate budget constraints and to accommodate a broader range of objectives.

\subsection{Our Results}
We study the design of near-optimal contracts in the multi-agent combinatorial-actions model of \cite{multimulti}, under budget constraints, for a broad class of objective functions.
In this setting, each agent $i \in [n]$ can select any subset from a given action set, and the principal is restricted by a budget constraint $B$, i.e., agents may receive at most a $B$-fraction of the reward. 
The case of $B=1$ corresponds to the unbudgeted setting. 
The class of objectives we consider---termed  \goodobj\ (BEyond STandard) by \cite{ feldman2025budget}---includes reward, welfare, profit, or any convex combination thereof (see \Cref{def:goodobj_multi_multi}). In what follows we present our main results (also illustrated in \Cref{fig:intersection-triangle}).

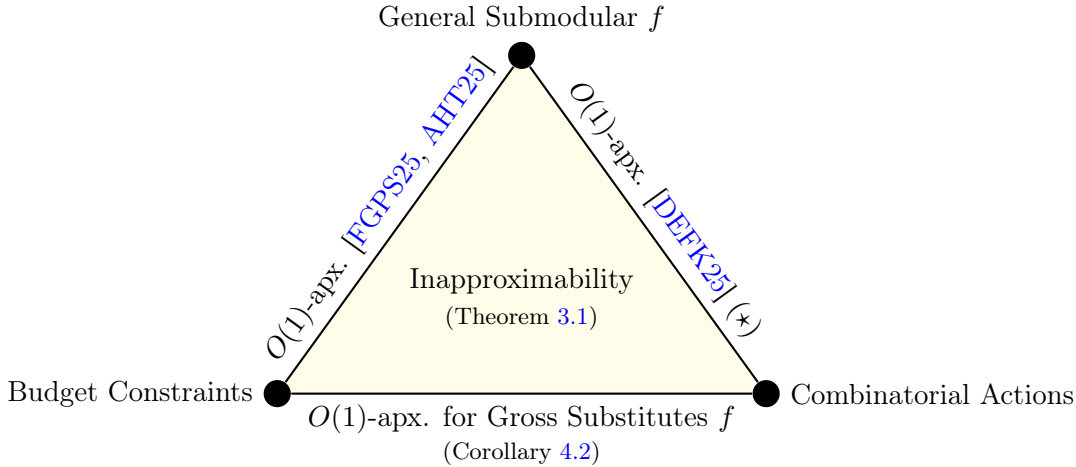
\begin{figure}
\begin{center}
\begin{tikzpicture}

  \fill[yellow!10] (-3.25,0) -- (3.25,0) -- (0,4.5) -- cycle;

  \node[fill, circle, minimum size = 5pt] (A) at (-3.25,0) {};
  \node[fill, circle, minimum size = 5pt] (B) at (3.25,0) {};
  \node[fill, circle, minimum size = 5pt] (C) at (0,4.5) {};
  \node (O) at (0,3) {};

  \node[anchor=east] at (A.west) {Budget Constraints};
  \node[anchor=west] at (B.east) {Combinatorial Actions};
  \node[anchor=south] at (C.north) {General Submodular $f$};
  \node at (0,1.5) {{Inapproximability}};
  \node at (0,1) {\footnotesize (\Cref{thm:multi-multi-inapprox})};

  \draw[thick] (A) -- (B) node[midway, below] {{$O(1)$-apx. for Gross Substitutes $f$}};
    \node[anchor=north] at (0,-0.5) {\footnotesize (\Cref{cor:gsapprox})};
  \draw[thick] (B) -- (C) node[midway, sloped, above] {$O(1)$-apx. \cite{multimulti} $(\star)$};
  \draw[thick] (C) -- (A) node[midway, sloped, above] {$O(1)$-apx. \cite{feldman2025budget,aharoni2025welfare}};

\end{tikzpicture}
\end{center}
\caption{
The three vertices of the triangle represent the dimensions along which the settings we consider differ: (i) the structure of $f$ (general submodular vs. gross substitutes), (ii) the presence of budget constraints, and (iii) the type of the agents' action space (binary vs. combinatorial).
Any pair of properties admits a constant-factor approximation, as indicated along each edge of the triangle. 
The figure illustrates that the impossibility arises only from the \emph{combination of all three} properties: the interior inapproximability region corresponds to submodular $f$ with budgets and combinatorial actions simultaneously.
$(\star)$ All results shown hold for all BEST 
objectives, except for the $O(1)$-approximation for the multi-agent combinatorial-actions setting 
without budget constraints of \cite{multimulti}, 
which holds only for profit maximization.
}
    \label{fig:intersection-triangle}
\end{figure}

\paragraph{Inapproximability for Submodular Instances.} 

Our first main result establishes a strong impossibility: when the reward function $f$ is submodular, there is no efficient approximation algorithm to any \goodobj\ objective.
This impossibility is information theoretic, and does not rely on any computational conjectures.

\begin{theoremnatural}[Inapproximability for Submodular Instances; \Cref{thm:multi-multi-inapprox}]\label{main_thm:inapproxourresults}
For the class of instances with submodular $f$, any \goodobj\ objective $\varphi$ (including profit, reward, and welfare),
any budget $B \in (0,1)$,
and any approximation guarantee $K : \mathbb{N} \to [1,\infty)$, any randomized algorithm can achieve a $K(n)$-approximation with respect to $\varphi$ under budget $B$ only with exponentially small probability in $n$, assuming it makes at most polynomially many value and demand queries to $f$.

\end{theoremnatural}
This impossibility is particularly striking in light of known positive results for multi-agent settings with submodular $f$: 
(i) constant-factor approximation to profit for  combinatorial actions without budget constraints \cite{multimulti}, and (ii) constant-factor approximation to any \goodobj\ objective in the binary-actions setting with budget constraints \cite{aharoni2025welfare,feldman2025budget}.
Our result shows that combinatorial actions and budget constraints \textit{together} rule out any approximation.
Moreover, our hardness construction constitutes the \textit{minimal} extension needed to obtain such an impossibility: All agents have binary actions, except for one agent which has one extra action.\footnote{In fact, we do not use the combinatorial nature of the agent's actions, just that there is more than one.}

Beyond establishing hardness of approximation, this result provides two important separations. First, between binary and combinatorial actions: Under budget constraints, the optimal-contract problem admits constant-factor approximation guarantees for binary actions but becomes inapproximable once agents have combinatorial actions. 
Second, between budgeted and unbudgeted settings: 
For combinatorial actions, profit maximization admits a constant-factor approximation without budgets, yet becomes inapproximable with budgets.

\paragraph{Constant-Factor Approximation for Gross Substitutes Instances.} On the positive side, our second main result shows that efficient constant-factor approximations are possible when $f$ lies in the class of gross substitutes (GS), an important and well-studied
subclass of submodular functions \cite{kelso1982job,gul1999walrasian,PaesLeme17}.
\begin{theoremnatural}[Constant-Factor Approximation for Gross Substitutes Instances; \Cref{cor:gsapprox}]\label{main_thm:GS}
For the class of instances with gross substitutes $f$ and any \goodobj\ objective $\varphi$ (including profit, reward, and welfare), there exists a deterministic polynomial-time algorithm that achieves an $O(1)$-approximation to $\varphi$ under any budget $B\in[0,1]$. This algorithm requires only value oracle access to $f$ and $\varphi$.
\end{theoremnatural}

Theorems~\ref{main_thm:inapproxourresults} and~\ref{main_thm:GS} together establish a separation between approximation guarantees under gross substitutes and general submodular functions.
Interestingly, a parallel (yet not directly related) separation is known for the problem of computing the \emph{exact} optimal contract in the single-agent combinatorial-actions setting: this problem admits a poly-time algorithm for GS rewards, whereas it is NP-hard for submodular rewards \cite{dutting2022combinatorial}.\footnote{
Whether a separation between GS and submodular rewards also appears in the multi-agent binary-actions setting remains an open problem: While constant-factor approximations are known for all submodular $f$ \cite{duetting2022multi}, it remains unknown whether better approximation (e.g., PTAS/FPTAS) can be achieved for GS. We note that computing the exact optimal contract is hard even for additive $f$.
}

Finally, we highlight two other aspects of tightness. Our inapproximability result applies even when the algorithm has access to both value and demand oracles, while our positive approximation guarantee relies only on value oracle access.\footnote{{This follows from the fact that when $f$ is GS, a demand query can be computed with poly-many value queries.}} Moreover, our hardness result applies to any randomized algorithm, whereas our proof for gross substitutes $f$ uses a deterministic algorithm.

\paragraph{Discussion: A Three-Way Barrier to Approximation.}
The multi-agent settings considered in this work differ along three dimensions: 
(i) the structure of $f$ 
(general submodular vs.\ gross substitutes), 
(ii) the presence of budget constraints, and (iii) the richness of the agents' action space
(binary vs.\ combinatorial).
Combined with previous literature, our results show that the inapproximability arises only from the \emph{combination of all three} properties, whereas any pair admits a constant-factor approximation. This is illustrated in \Cref{fig:intersection-triangle}.

\paragraph{FPTAS for Special Cases.} 
We give an FPTAS for two special cases of interest: instances with additive $f$ and instances with a single agent. 

\begin{theoremnatural}[FPTAS for Additive Instances; \Cref{thm:fptas}]\label{main_thm:fptas}
For any budget $B \in [0,1]$, each of the objectives profit, reward, and welfare admits a deterministic FPTAS with value queries, when $f$ is additive.
\end{theoremnatural}

This result establishes the first FPTAS for the multi-agent combinatorial-actions model, even without budgets. It generalizes the FPTAS of \cite{duetting2022multi,feldman2025budget} for the multi-agent binary-actions model, where the latter also considers budgets.
We note that, whereas our previous results apply to all \goodobj\ objectives, \Cref{main_thm:fptas} applies only to profit, reward, and welfare.

We also generalize the FPTAS of \cite{multimulti} for maximizing profit in the \textit{single-agent} setting with combinatorial actions to the budgeted setting.\footnote{Maximizing welfare and reward can be achieved by paying the agent the entire budget upon success.}
\begin{theoremnatural}[FPTAS for Single-Agent Instances; \Cref{thm:single-agent-fptas}]
    In the single-agent setting, there exists an FPTAS for maximizing the principal's profit under any budget $B \in [0,1]$, when $f$ is monotone. This algorithm requires access to a demand oracle to $f$.
\end{theoremnatural}

\subsection{Our Techniques}

\paragraph{Inapproximability for Submodular Instances.} A key property in the binary-actions setting with submodular $f$ is the \emph{best-response monotonicity}: given a contract and a Nash equilibrium $S$, 
any agent in $S$ continues to find effort a best response, even if a subset of the remaining agents in $S$ shirk.
This property is central to achieving constant-factor approximations in the binary-actions setting, both with and without budget constraints \cite{duetting2022multi,aharoni2025welfare,feldman2025budget}. In contrast, best-response monotonicity \emph{fails} in the combinatorial-actions setting with submodular $f$: under a fixed contract and an equilibrium, if some agents reduce their action sets, the best response of an agent $i$ may also be to reduce his; see, e.g., \cite[Example 1.1]{multimulti}. 
We identify non-monotonicity of the best response as a core aspect of the complexity of the combinatorial-actions setting and our proof of \Cref{main_thm:inapproxourresults} leverages a carefully engineered violation of this property.

Our construction consists of $n-1$ binary-action (work or shirk) agents and one special agent with two non-shirking actions: a good action $\mathcal{G}$ and a bad action $\mathcal{B}$. The reward function is at least ${1}/{2}$ when the special agent takes $\goodaction$, and infinitesimally small otherwise.
Incentivizing the special agent to take $\mathcal{G}$ alone, however, would necessarily violate the budget constraint.
By carefully constructing the reward function $f$, we ensure that there exists a budget-feasible contract that incentivizes the special agent to take $\goodaction$, but only if it simultaneously incentivizes a unique subset $A'$ of exactly half of the remaining agents to exert effort. 
Since $\mathcal{G}$ can be incentivized only in combination with other agents, this constitutes a violation of best-response monotonicity. 
The ``good'' equilibrium---$A'$ together with $\goodaction$---yields a profit of at least $(1-B) \cdot ({1}/{2})$, while any other budget-feasible contract fails to incentivize the action $\goodaction$ and thus achieves only an infinitesimally small reward.
Therefore, an algorithm must incentivize the ``good'' equilibrium to achieve any desirable approximation. 
Yet the number of candidate subsets $A'$ is exponential, and we prove that no efficient method can identify $A'$ using poly-many value or demand queries, establishing our inapproximability result.

Crucially, this construction breaks down for profit maximization in the unbudgeted setting ($B=1$), where \cite{multimulti} provide an approximation algorithm. Indeed, when $B=1$, the ``good'' equilibrium in our construction yields a profit of $(1-B)\cdot ({1}/{2}) = 0$, nullifying the argument.

\paragraph{Constant-Factor Approximation for Gross Substitutes Instances.} 
The first step toward establishing our positive result is {proving} that 
when $f$ is gross substitutes, best-response monotonicity holds even when agents have combinatorial actions. 
The key idea in the proof of this observation is  to connect  agents' best responses with demand bundles (see \Cref{sec:model} for formal definitions). 
More precisely, we observe that agent $i$'s best response for a given profile of actions of the other agents is a demand bundle with respect to a carefully chosen price vector.
We then show that when the set of actions taken by the other agents is reduced to any subset, the corresponding price vector increases coordinate-wise. Therefore, by the gross-substitutes property, agent $i$ is incentivized to take a (weak) superset of the original actions, establishing best-response monotonicity.

We find that best-response monotonicity, combined with tools from \cite{multimulti}, 
yields
an {up-to-a-constant-factor} equivalence between any two \goodobj\ objectives under any two budgets,
paralleling the result of \cite{feldman2025budget} for the binary-actions case. This equivalence, together with the $O(1)$-approximation for profit from \cite{multimulti}, implies a constant-factor approximation for any \goodobj\ objective under any budget.

\subsection{Related Work}\label{subsec:RelatedWork}

\paragraph{Combinatorial Contracts.}
A combinatorial model for contracting agents with binary actions was introduced by \cite{babaioff2006combinatorial,BabaioffFNW12}, where the principal selects a subset of agents to incentivize, and the reward function $f$ is Boolean. This foundational model was later extended to allow mixed strategies \cite{babaioff2006mixed}, and subsequent work studied free-riding in this setting \cite{babaioff2009free}.

\cite{duetting2022multi} generalized this framework by considering richer reward functions drawn from the complements-free hierarchy \cite{lehmann2001combinatorial}, showing that for XOS functions, constant-factor approximations are achievable using demand and value queries. In subsequent work, the same authors extended this to a multi-action setting \cite{multimulti}, where agents choose arbitrary subsets of actions, and demonstrated that for submodular $f$, poly-many demand queries suffice for approximation.
\cite{alon2025multi} considered a multi-project setting where each project corresponds to a binary-action problem, and showed constant-factor approximations for XOS valuations. \cite{castiglioni2025fair} examined project allocation under fairness constraints, proving inapproximability of envy-freeness in the general case and tractability in constant-size markets.

In the single-agent variant, \cite{dutting2022combinatorial} introduced a model where the agent selects subsets of $n$ costly actions, and the principal's reward is given by a set function $f$. They showed tractability for gross-substitutes functions using value queries. Follow-up works \cite{deo2024supermodular, dutting2024query, ezra2023Inapproximability, contractsBeyondGS, feldman2025ultraefficient} further refined the complexity landscape in this combinatorial-action setting.
Another flavor of the single agent setting, studied in \cite{contractsSequential}, considers a scenario where the agent takes actions sequentially, rather than simultaneously, observing intermediate outcomes. They provide algorithms and hardness results for independent and correlated actions.

The recent work of \cite{contractsInspection} relaxes the hidden-action assumption by introducing a combinatorial inspection model, where the principal can pay to verify whether the agent's chosen action belongs to a given set, and the inspection cost is combinatorial.

\paragraph{Linear Contracts.}
Linear contracts are those in which the principal pays each agent a fixed fraction of the reward.
The robustness of linear contracts was first established by \cite{carroll2015robustness}, who proved that linear contracts are max-min optimal when the principal is unaware of the agent's full action set. This result was later generalized to randomized actions and contracts by \cite{peng2024optimal}. 
Robustness was also established in a different model of uncertainty by \cite{dutting2019simple}, who additionally provided approximation guarantees for linear contracts relative to the optimal ones.
Another desirable trait of linear contracts is ambiguity-proofness. That is, unlike general contracts, the principal cannot benefit from introducing ambiguity into linear contracts \cite{dutting2024ambiguous, duetting2025succinct}.

\paragraph{Contracts under Budget Constraints.}
There has been growing interest in the study of budgets in algorithmic contract design.
\cite{hann2024optimality} studied the budgeted multi-agent setting where each agent performs an independent task with an observable (binary) outcome. 
\cite{doron2025algorithm} studied a setting where the agent has a budget constraint and show how to convert algorithms for the agent's utility maximization problem to an approximation algorithm with multiplicative and additive guarantees.

Closer to our setting, \cite{feldman2025budget} and \cite{aharoni2025welfare} introduced budget constraints into the binary-action model of \cite{duetting2022multi}. \cite{feldman2025budget} introduced the notion of \goodobj\ objectives and showed that, via a downsizing lemma, any two \goodobj\ objectives and any two budgets are equivalent up to a  constant factor. As a result, their work also yields constant-factor approximations for XOS and submodular objectives via demand and value queries, respectively. \cite{aharoni2025welfare} focused on social welfare maximization and studied its divergence from the principal's profit. They have also given constant factor approximation algorithm for the social welfare when the budget $0 \le B \le 2-\Theta(1)$.

\paragraph{Other Contractual Models.}
Non-combinatorial multi-agent settings with agent externalities are studied in \cite{segal1999contracting,segal2003coordination,bernstein2012contracting}. 
A complexity measure for the implementation of multi-agent contracts was proposed by \cite{babaioff2014contract}. 
More recently, \cite{cacciamani2024multi,castiglioni2023multi} have studied multi-agent contracts beyond binary outcomes. 
Contracts where agents have private types are studied in \cite{alon2021contracts, alon2022bayesian, castiglioni2025reduction, CastiglioniM021, GuruganeshSW023, guruganesh2021contracts, castiglioni2022designing}. 
Several recent works explore the intersection of contracts and learning \cite{ZhuBYWJJ23, cohen2023learning, ho2014adaptive, BacchiocchiC0024,chen2024boundedcontractslearnableapproximately, duetting2025pseudodimensioncontracts}.

\section{Model and Preliminaries}\label{sec:model}

In this section we present the multi-agent combinatorial-actions model of \cite{multimulti}, the family of \goodobj\ (BEyond STandard) objectives of \cite{feldman2025budget}.

\subsection{The Multi-Agent Combinatorial-Actions Setting}
We consider a contractual setting involving a single principal and a set $A$ of $n$ agents.
The principal is the owner of a project, which can either succeed or fail.
Upon success, the project yields a reward of $1$ for the principal, and $0$ otherwise.
The success probability of the project is determined by the (costly) actions that the agents perform, as specified below. 
The principal does not observe these actions, only the outcome. In order to incentivize the agents, she offers each $i \in A$ a \textit{linear contract} $\icontract \in [0,1]$ that specifies their payment upon success, while the payment upon failure is always $0$.\footnote{
One can also view this model as a multiple-outcome project, where the principal only observes the reward and is restricted to using only linear contracts.}
As we formally establish in \Cref{sec:linear}, for all objectives considered in this paper, it is without loss of generality to restrict attention to linear contracts.

Each agent $i \in A$ has a set of actions $T_i$ from which they may choose any subset. In particular, $T_i = \{j\}$ corresponds to the binary-actions case \cite{duetting2022multi, feldman2025budget}, as the agent may choose either $j$ or $\emptyset$.
We assume that any two agents have disjoint sets of actions, i.e., $T_i \cap T_{i'} = \emptyset$ for any $i \ne i'$.
Each action $j \in T_i$ is associated with a non-negative cost $c_j \geq 0$, and the cost of a set of actions is additive, i.e., $c(S_i) = \sum_{j \in S_i} c_j$ for any $S_i \subseteq T_i$.
In particular, $c(\emptyset) = 0$.
The disjoint union of all possible actions is denoted by $T = \bigsqcup_{i \in \agents} T_i$.
Given a profile of actions $S\subseteq T$, we denote by $S_i$ the actions in $S$ taken by agent $i$, i.e., $S_i = S\cap T_i$. Similarly, we use $S_{-i} = \bigsqcup_{i' \in A \setminus \{i\}} S_{i'}$ to denote the actions taken by all agents \textit{except for} agent $i$, in the profile of actions $S$.

As mentioned, we consider a binary-outcome project that yields a reward of 1 to the principal upon success and 0 otherwise.
{At the heart of the model is a} monotone combinatorial set function $f: 2^T \to [0,1]$, which maps every profile of actions, $S = \bigsqcup_{i \in A} S_i$, chosen by the agents to the project's success probability. Since the principal's reward upon success is normalized to $1$, $f(S)$ also represents the principal's expected reward. We assume that $f(\emptyset) = 0$.

A problem instance is a tuple $\multiInstance$. 
Here, $A$ is the set of agents, $T_i$ denotes the set of available actions for each agent $i \in A$, $f$ is the reward function, and $c = \{c_j\}_{j \in T}$ represents the costs associated with each action.
{Recall that} the principal offers a (linear) contract $\noindexcontract \in [0,1]^\agents$, where $\icontract$ is the transfer from the principal to agent $i \in \agents$, if the project succeeds.

\paragraph{Utilities and Equilibria.}
For a given contract $\noindexcontract$ and a profile of actions $S = (S_1,\dots,S_n)$, agent $i$'s expected utility is given by $\icontract \cdot f(S_i \sqcup S_{-i})-c(S_i)$, i.e., by the expected payment minus cost.
We say that a contract $\noindexcontract$ \textit{incentivizes} the profile $S$ if $S$ forms a Nash equilibrium (NE) under $\noindexcontract$. Namely, for every agent $i \in A$ and every alternative set of actions $S'_i \subseteq T_i$, the following holds:
$$
\icontract \cdot  f(S_i \sqcup S_{-i}) - c(S_i) \ge \icontract \cdot  f(S'_i \sqcup S_{-i}) - c(S'_i).
$$
It was shown in \cite{multimulti} that every contract $\noindexcontract$ incentivizes at least one pure Nash equilibrium.
A contract $\noindexcontract$ may incentivize multiple equilibria, and we denote this collection by $\nash(\noindexcontract)$.

Given a contract $\noindexcontract$ and an equilibrium $S \in \nash(\noindexcontract)$, the principal's utility{, or \textit{profit},} is defined to be the  expected reward minus payment, i.e.,
$$
u_P(\noindexcontract,S) = \left(1 - \sum_{i\in \agents} \icontract\right)f(S).
$$

\paragraph{{Maximizing Profit Under Budget Constraints.}}
Given a budget $B$, we call a contract $\noindexcontract$ \emph{budget-feasible} if $\sum_{i\in \agents} \icontract \le B$.
For convenience, we denote the set of all pairs of budget-feasible contracts and incentivized equilibria\footnote{A contract may incentivize multiple equilibria, in which case it will appear in $\setofcontracts$ more than once.} by $\setofcontracts$, namely,
$$
\setofcontracts = \{ (\noindexcontract,S) \mid \sum_{i\in \agents} \icontract \le B \text{ and } S \in \nash(\noindexcontract) \}.
$$
The standard goal of contract design is to maximize the profit, $u_P(\noindexcontract,S)$, which we use to illustrate the budgeted contract design problem:
Given a budget $B$ we seek a contract and an equilibrium which approximately maximize the profit., i.e., $(\noindexcontract,S) \in \setofcontracts$ such that
$\gamma \cdot u_P(\noindexcontract,S) \ge \max_{(\noindexcontract',S') \in \setofcontracts} u_P(\noindexcontract',S')$, for some constant $\gamma \ge 1$. In this work, we also consider additional important objectives, beyond profit maximization, which we elaborate on below.

In light of the observation made in \cite[Proposition 3.13]{feldman2025budget}, we cannot guarantee to find a contract $\noindexcontract$ for which \emph{every} incentivized equilibrium $S \in \nash(\noindexcontract)$ yields a $\gamma$-approximation to the optimal profit, so we must settle for a contract $\noindexcontract$ for which \emph{some} incentivized equilibrium yields a $\gamma$-approximation.
Therefore, our algorithms output a contract along with an equilibrium which provide the desired guarantees.

\paragraph{Subset Stability and the Doubling Lemma.}
The notion of \emph{subset stability} and the \emph{doubling lemma} play a crucial role in the analysis of the algorithm of \cite{multimulti}, which computes a near-optimal contract with respect to profit without budget constraints.
Subset stability is a relaxation of the Nash equilibrium condition: A profile of actions $S = \bigsqcup_{i \in \agents} S_i$ is subset-stable with respect to a contract $\noindexcontract$ if no agent $i \in \agents$ strictly benefits from deviating to a subset of $S_i$.
\begin{definition}[Subset Stability, Definition 3.2 of \cite{multimulti}]\label{def:subsetstable}
    A set of actions $S$ is subset-stable with respect to contract $\noindexcontract$, if for every agent $i$, every subset of his actions $S'_i \subseteq S_i$ satisfies
    $$
    \icontract \cdot f(S_i \sqcup S_{-i}) - c(S_i) \ge \icontract \cdot f(S'_i \sqcup S_{-i})- c(S'_i).
    $$
\end{definition}

The doubling lemma shows how a subset-stable profile $S$ with respect to a contract $\noindexcontract$ can be used to incentivize a Nash equilibrium that guarantees at least a half of the expected reward.
\begin{lemma}[Doubling Lemma, Lemma 3.3 of \cite{multimulti}]\label{lem:doubling}
    Let $f$ be a submodular function and let $\eps > 0$. 
    If $S$ is a subset-stable set of actions with respect to a contract $\noindexcontract$, then any equilibrium $S'$ with respect to the contract $2\noindexcontract + \epsilonvec$ satisfies $f(S') \ge (1/2) \cdot f(S)$, where $\epsilonvec = (\eps,\dots,\eps)$.
\end{lemma}

\paragraph{Restricted Contracts.}
For any contract $\noindexcontract\in \reals^n_+$ and any set of agents $G \subseteq A$,  we denote by $\noindexsubcontract{G}$ the contract obtained by restricting payments to the agents in $G$. Namely, we let
$$\noindexsubcontract{G}=\begin{cases}
\icontract & \text{if } i \in G, \\    0 & \text{otherwise}.
\end{cases}$$
When $G=\{i\}$ is a singleton, we often omit the brackets and write $\noindexsubcontract{i}$.

\paragraph{Reward Functions and Access Oracles.}
In this work we focus on monotone reward function $f$.
We denote by $f(S' \mid S)$ the marginal contribution of $S'$ to $f$, given $S$, i.e.,
$f(S' \mid S) = f(S' \cup S) - f(S)$.
We also use the notation $f_S(i) = f(\{i\} | S)$.
We consider reward functions that belong to one of the following (nested) classes:
\begin{itemize}
    \item Additive: $f(S) = \sum_{a \in S} f(\{a\})$ for any $S \subseteq T$.
    \item Gross-Substitutes: For any two price vectors $p,q \in \reals^{|T|}$, such that $p \le q$ coordinate-wise, and any $S^\star \in \argmax_{S \subseteq T} \{ f(S) - \sum_{a \in S} p_a \}$, there exists a \textit{demand bundle} with respect to $q$, i.e., $S' \in \argmax_{S \subseteq T} \{ f(S) - \sum_{a \in S} q_a \}$, such that $S^\star \cap \{j \mid p_j = q_j\} \subseteq S'$.
    \item Submodular: For any $S' \subseteq S \subseteq T$ and any $a \notin S$, it holds that
    $f{_{S'}}(a) \ge f_{{S}}( a )$.
    \item Subadditive: For any $S', S \subseteq T$, it holds that $f(S) + f(S') \ge f(S \cup S')$.
\end{itemize}
It is well-known that,
$\text{additive} \subsetneq \text{gross-substitutes} \subsetneq \text{submodular} \subsetneq \text{subadditive}$ \cite{lehmann2001combinatorial}.
As $f$ may have an exponential representation in $n$ (and $|T|$) we assume {algorithms} access $f$ via two standard models:
\begin{itemize}
    \item Value oracle: accepts a set $S \subseteq T$ and returns $f(S)$.
    \item Demand oracle: accepts a vector $p \in \reals^{|T|}$ and returns  $S^\star \in \argmax_{S \subseteq T} \{ f(S) - \sum_{a \in S} p_a \}$.
\end{itemize}
Generally, demand oracles are strictly stronger that value oracles \cite{blumrosen2005computational}. However, for gross-substitutes $f$ a demand query can be computed in poly-time with value oracle access \cite{PaesLeme17}.

\subsection{Objectives and Maximization Problems}\label{sec:best_multi}

Our results apply to a variety of objectives beyond the standard goal of maximizing profit.
Analogously to \cite{feldman2025budget}, we consider a family of natural objectives that are defined with respect to a contract $\noindexcontract$ and the set of actions $S=\bigsqcup_{i \in A} S_i$ taken by the agents.
We begin by presenting a precise definition of an objective.

\begin{definition} [Objectives in the Multi-Agent Combinatorial-Actions Model]\label{def:obj_multi_multi}
    An \emph{objective} $\varphi$ is defined by a poly-time algorithm that, given a problem instance $\multiInstance$, a contract $\noindexcontract$, and a subset of actions $S\subseteq T$, outputs a non-negative real number, denoted $\varphi_{\multiInstance}(\noindexcontract,S)$. This algorithm has value oracle access to $f$. 
    We omit the subscript if the instance is clear from context.
\end{definition}

Observe that the standard goal of maximizing profit adheres to this definition: Given a contract $\noindexcontract$ and a set of actions $S$ a simple poly-time algorithm can compute $u_P(\noindexcontract,S) = (1-\sum_{i \in A} \icontract) \cdot f(S)$ using a single value query to $f$.

For a given objective $\varphi$ and budget $B\in (0,1]$, we denote by $\MAX\varphi(B)$ the problem of finding a budget-feasible contract and equilibrium $(\noindexcontract,S) \in \setofcontracts$, that maximize $\varphi$. We sometimes abuse notation and use $\MAX\varphi(B)$ to denote the maximal value of $\varphi$ achievable under budget $B$.

Below we define the class of \goodobj\ objectives, which naturally extends the definition for the binary-actions setting of \cite{feldman2025budget}.

\begin{definition}[{BEST Objectives}]  \label{def:goodobj_multi_multi}
    An objective $\varphi$ belongs to the class of  beyond standard (BEST) objectives if, for any instance $\multiInstance$, it is:
    \begin{enumerate}[label=(\roman*)]
        \item \emph{Sandwiched between profit and reward:} For any $\noindexcontract$ and $S \subseteq \actions$,
        $u_P(\noindexcontract,S)\le \varphi(\noindexcontract, S) \le f(S)$. \label{def:sandwichproperty}
        \item \emph{Decomposable:} For any $\noindexcontract$, any $S \subseteq \actions$, and any $i\in \agents$,  $\varphi(\noindexcontract, S) \le f(S_{-i})+\varphi(\noindexsubcontract{i}, S_i)$.\label{def:decompositionproperty} 
        \item \emph{Weakly increasing in $S$:} For any $\noindexcontract$ and any $S\subseteq S'\subseteq \actions$, $\varphi(\noindexcontract, S) \le \varphi(\noindexcontract, S')$. \label{def:weakly_increasing_property}
        \item \emph{Weakly decreasing in $\noindexcontract$:} For any  $\noindexcontract \le \noindexcontract'$ (coordinate-wise) and $S\subseteq \actions$,  $\varphi(\noindexcontract, S) \ge \varphi(\noindexcontract', S)$. \label{def:weakly_decreasing_property}
    \end{enumerate}
\end{definition}
The first two conditions generalize the definition of \goodobj\ objectives from \cite{feldman2025budget}.
The richer combinatorial-actions setting considered here also requires conditions $(iii)$ and $(iv)$, which are natural, and clearly hold for common objectives such as profit, welfare and reward.

In \Cref{sec:proofs_BEST} we show that the class of \goodobj\ objectives includes the standard objectives of profit (maximizing the principal's utlity), welfare ($\max_{\noindexcontract, S \in \nash(\noindexcontract)} f(S)-c(S)$) and reward ($\max_{\noindexcontract, S \in \nash(\noindexcontract)} f(S)$); see \Cref{obs:common_objs}. We also show that this class is closed under convex combinations; see \Cref{obs:convex_comb_objs}.

\section{Inapproximability Result for Submodular Instances}\label{sec:inapprox_submodular}
In this section we show that when $f$ is submodular, for any budget $B\in (0,1)$, and any \goodobj\ objective $\varphi$, it holds that $\MAX\varphi(B)$ cannot be efficiently approximated with demand query access. This presents a significant departure from the unbudgeted case, where \cite{multimulti} present a poly-time constant-factor approximation algorithm for profit, i.e., for $\maxrevenue(1)$. 

\begin{theorem}[{Inapproximability for Submodular Instances}]\label{thm:multi-multi-inapprox}
    Fix any \goodobj\ objective $\varphi$ and any budget $B\in (0,1)$.
    For any approximation guarantee $K:\mathbb{N} \to [1,\infty)$, any (randomized) poly-time algorithm with demand oracle access to $f$ may only achieve a $K(n)$-approximation to $\MAX\varphi(B)$ with exponentially-small probability (in $n$).
\end{theorem}

By setting $K(n) = 2^n$ in \Cref{thm:multi-multi-inapprox}, we derive the following corollary.
\begin{corollary}[{Exponential Lower Bound on Expected Approximation Ratio}]
    Fix any $\goodobj$ objective $\varphi$, and any budget $B\in (0,1)$. 
    Let $\mathcal{A}$ be any (randomized) poly-time algorithm with demand oracle access to $f$, and denote its output by $(\noindexcontract,S)$.
    Then, there exists a family of instances such that
    ${\MAX\varphi(B)}/{\mathbb{E}_{\mathcal{A}}[\varphi(\noindexcontract,S)]} = 2^{\Omega(n)}$.
\end{corollary}

Our result is information-theoretic: We construct an instance with a probability distribution over submodular reward functions and establish an upper bound on the expected performance of any deterministic algorithm with demand oracle access on this randomized input. By Yao's principle, this directly implies the statement of the theorem.

Let us now define the instances used in the proof of \Cref{thm:multi-multi-inapprox}.

\begin{definition}[Parameterized Instances]
Fix a budget $B\in (0,1)$, an approximation guarantee $K:\mathbb{N} \to [1,\infty)$, and any even $n>0$. Fix any $\varepsilon > 0$ such that:
\begin{align*}
\varepsilon <  \min\left( \frac{1-B}{K(n) \cdot (n+4)}, \frac{4 \cdot n}{B}\right).
\end{align*}
For any $A' \subseteq [n]$ with $|A'| = n/2$, we define an instance $\mathcal{I}^{(A')} = \langle A, \bigsqcup T_i, f^{(A')}, c \rangle$ as follows.
\begin{itemize}
        \item The set of agents is $A = \{1, \ldots, n+1\}$. 
        \item  Each agent $i\in [n]$ controls a single action $\actions_i =\{i\}$. Additionally, agent $n+1$ controls two actions $\actions_{n+1} = \{\badaction, \goodaction\}$.
That is, the total set of actions is $\actions = [n] \cup \{\badaction, \goodaction\}$, ($\badaction$ for ``bad'', $\goodaction$ for ``good'').
\item The costs of the actions are:
\begin{align*}
    c_i = \eps^3 \text{ for all } i\in [n], \quad c_{\goodaction} = (1/2) \cdot \left(B-(n/2) \cdot \eps^2\right) \quad\text{and}\quad c_{\badaction}= (3/2) \cdot \eps \cdot B.
\end{align*}
\item The reward function $f^{(A')}$ is defined as a sum of three set functions. Specifically, we let $f^{(A')}(S) = f_1(S) + f_2(S) - f_3(S, A')$, where:
    \begin{align*}
            f_1(S) &= \max\left((1/2) \cdot \indicator[\goodaction\in S], \eps \cdot \indicator[\badaction\in S]\right)\\
            f_2(S) &= \eps \cdot \min \left(|S \setminus \{\goodaction\}|, n/2+1\right)\\
            f_3(S, A') &= (\eps/2) \cdot \indicator[S=
        \{\badaction\} \cup \specialAgents]
        \end{align*}
    \end{itemize}
\end{definition}
Whenever it is clear from the context, we omit $A'$ from the reward function, and write simply $f$.
Observe that $f_1$ is a (weighted) unit-demand function over $\{\goodaction,\badaction\}$ and $f_2$ is a uniform $({n}/{2}+1)$-demand\footnote{{A set function $f: 2^A \to \mathbb{R}_{\geq 0}$ is \emph{unit-demand} if
$f(S) =\; \max_{i \in S} f(\{i\})$. Moreover,
$f$ is \emph{uniform $k$-demand} if there exists $v\in \reals_{\ge 0}$ such that
$f(S) = \min\{\,|S|,\, k\,\} \cdot v$}} over $[n] \cup \{\badaction\}$.

We give a brief proof sketch for \Cref{thm:multi-multi-inapprox}:
First, we observe that obtaining a good approximation for any \goodobj\ objective requires a good approximation to $f$. 
By design, this is only possible when agent $n+1$ chooses action $\goodaction$, since any set containing $\goodaction$ has value greater than $1/2$, whereas any set that excludes $\goodaction$ has value at most $(n/2 + 2) \cdot \varepsilon$. 
However, the only way to incentivize agent $n+1$ to take $\goodaction$ instead of $\badaction$, while complying with the budget, is to incentivize the set $\specialAgents$ to take action. 
This is because the marginal reward of $\badaction$ is reduced when the agents of $\specialAgents$ exerts effort 
(as captured by $f_3$).
Consequently, a good approximation can only be achieved when the equilibrium satisfies $S_{-(n+1)} = \specialAgents$, implying that the algorithm must effectively ``know'' the set $\specialAgents$. 
In the case of value queries alone, a standard ``hide a special set'' argument shows that any algorithm identifying $\specialAgents$ with non-negligible probability must make exponentially many value queries. 
Finally, to complete the proof, we demonstrate that access to demand queries does not help, as any demand query can be simulated using $O(1)$ value queries for our choice of $f$.

We now move to the formal analysis. Before proving \Cref{thm:multi-multi-inapprox}, we present some useful lemmas.
We first show that $f^{(A')}$ is monotone and submodular. The proof is deferred to \Cref{app:missing_proofs}.

\begin{restatable}{lemma}{monotoneAndSubmodular}\label{lem:monotone_and_submodular}
    For any $A' \subseteq [n]$ with $|A'| = n/2$, it holds that $f^{(A')}$ is monotone and submodular.
\end{restatable}

We next show that demand queries to $f^{(A')}$ are not more powerful than value queries.

\begin{lemma} \label{lem:demand_query_by_values}
For  any $A' \subseteq [n]$  with $|A'| = n/2$, any demand query to $f^{(A')}$ can be computed with $12$ value queries to $f^{(A')}$.
\end{lemma}
\begin{proof}
    Fix a subset $A' \subseteq [n]$ with $|A'| = n/2$.
    Let $p\in \reals_{\ge 0}^{n+1}$ be a price vector. Without loss of generality, assume that $p_1 \le \dots \le p_n$. Let $k$ be a maximal index such that $p_k < \eps$ or $0$ if no such index exists. Let $\tau = \min\{k,n/2+1\}$.
    
    We claim that one of the sets $\{\{1,\dots,\tau\}, \{1,\dots,\tau-1\}, \{1,\dots,\tau-2, \tau\}\}$ combined with one of $\{\emptyset, \{\goodaction\}, \{\badaction\}, \{\goodaction, \badaction\}\}$ is a demand bundle.
    Once the claim is proven, it follows that $12$ value queries suffice to answer a demand query.
    
    To prove the claim, let $S^\star\in \argmax_{S \subseteq T} \{f(S) - \sum_{j \in S} p_j \}$, be a set in the demand with respect to price vector $p$, and let $S_0 = S^\star\cap \{\goodaction,\badaction\}$. 
    The marginal utility 
    of $S' \subseteq [n]$ with respect to $S$ is:
    \begin{align*}
        u(S'\mid S_0) 
        &=
        f(S'\mid S_0) - \sum_{i\in S'}p_i \\
        &=\begin{cases}
        \varepsilon \cdot \min (|S'|, n/2 + 1)-\sum_{i\in S'} p_i & 
        \text{if } S_0 = \emptyset
        \text{ or }S_0 = \{\goodaction\} \\ 
        \varepsilon \cdot \min (|S'|, n/2)-\sum_{i\in S'} p_i & \text{if } S_0 = \{\badaction, \goodaction\} \\
        \varepsilon \cdot \min (|S'|, n/2)-(\eps/2) \cdot \indicator[S'=\specialAgents] - \sum_{i\in S'} p_i &  \text{if } S_0 = \{\badaction\} 
    \end{cases}
    \end{align*}
    We claim that one of $\{\{1,\dots,\tau\}, \{1,\dots,\tau-1\}, \{1,\dots,\tau-2, \tau\}\}$ maximizes $u(S' \mid S_0)$.    
    
    Indeed, if $S_0 = \emptyset$ or $S_0 =\{\goodaction\}$, clearly $S'=\{1,\dots,\tau\}$ maximizes $u(S'\mid S_0)$, and we are done. If $S_0 = \{\badaction, \goodaction\}$, clearly $S=\{1,\dots, \min(k, n/2)\}$ maximizes $u(S' \mid S_0)$, and by definition of $\tau$, we have $\min(k, n/2)\in\{\tau, \tau-1\}$, implying that either $\{1,\dots, \tau\}$ or $\{1,\dots, \tau-1\}$ maximizes $u(S' \mid S_0)$, as needed.
    Otherwise, $S_0 = \{\badaction\}$.
    If $\{1,\dots \tau-1\} = \specialAgents$, then since $|A'| = n/2$, we must have  $\tau={n}/{2}+1$, and so  either
    $\{1,\dots,\tau -2 , \tau\}$ or $\{1,\dots \tau-1\}$ maximizes $u(S' \mid S_0)$.
    If $\{1,\dots \tau\} = \specialAgents$, then since $|A'| = n/2$, we must have  $p_{\tau+1} \ge \eps$.
    We get that either $\{1,\dots, \tau-1\}$ or $\{1,\dots, \tau\}$ maximizes $u(S' \mid S_0)$.
    If $\{1,\dots \tau-1\} \ne \specialAgents$ and $\{1,\dots \tau\} \ne \specialAgents$, then, either $\{1,\dots, \tau-1\}$ or $\{1,\dots, \tau\}$ maximizes $u(S' \mid S_0)$.
\end{proof}

Given the lemma above, it is enough to show that no algorithm can obtain a good approximation using only value queries.

The following lemma shows that it is possible to incentivize the set of actions $A' \cup {\goodaction}$ without violating the budget constraints.

\begin{lemma} \label{lem:instance_equilibriums}
    For any instance $\mathcal{I}^{(A')}$ with $A' \subseteq [n]$ and $|A'| = n/2$, there exists a contract $\noindexcontract$ such that $(\noindexcontract, \specialAgents \cup \{\goodaction\})\in \setofcontracts$. 
\end{lemma}
\begin{proof}
Consider the contract 
    $$
    \icontract = \begin{cases}
        \eps^2 & \text{if } i \in \specialAgents \\
        B-(n/2) \cdot \varepsilon^2 & \text{if } i = n+1 \\
        0 & \text{otherwise},
    \end{cases}
    $$
    and let $S=\specialAgents\cup \{\goodaction\}$. 
    Clearly $f(S) \ge f(\goodaction) \ge 1/2$, we now show that $S\in \nash(\noindexcontract)$.
    Let $j\in [n] \setminus \specialAgents$. Since $\someindexcontract_j = 0$, $S_j = \emptyset$ is a best response for agent $j$.
    Let $i\in \specialAgents$. Note that
    \[
    \icontract \cdot f(i\mid S_{-i}) = \eps^2 \cdot f_2(i\mid S_{-i}) = \eps^2\cdot \varepsilon=c_i,
    \]
    so $S_i = \{i\}$ is a best response for agent $i$ since his only choices are $\{i\}$ and $\emptyset$.
    We now turn to agent $n+1$. Note that 
    \[
    \someindexcontract_{n+1} \cdot f(\badaction \mid S_{-(n+1)}) = \someindexcontract_{n+1}\cdot f(\badaction\mid \specialAgents) = \left(B-(n/2) \cdot \eps^2\right)\cdot (3/2) \cdot (\eps) < c_{\badaction},
    \]
    and so since $f$ is submodular, it holds that agent $(n+1)$'s best response does not contain $\badaction$. It therefore remains to show that agent $(n+1)$'s utility from $\goodaction$ is non-negative. Indeed,
    \[
    \someindexcontract_{n+1} \cdot f(\goodaction \mid S_{-(n+1)}) = \someindexcontract_{n+1}\cdot f(\goodaction\mid \specialAgents) = \left(B- ({n}/{2}) \cdot \eps^2\right) \cdot ({1}/{2}) = c_{\goodaction},
    \]
    as needed.
\end{proof}

Next, we show that incentivizing $A' \cup {\goodaction}$ is necessary to get a non-trivial approximation to $f$.

\begin{lemma} \label{lem:instance_equilibriums_upper_bound}
    For any instance $\mathcal{I}^{(A')}$ with $A' \subseteq [n]$ and $|A'| = n/2$, for any $(\noindexcontract, S) \in \setofcontracts$ with $S\ne \{\goodaction\}\cup \specialAgents$, it holds that $f(S)\le(n/2+2) \cdot \eps$.
\end{lemma}
\begin{proof}
    Fix a subset $A' \subseteq [n]$ with $|A'| = n/2$.
    Let $(\noindexcontract,S)$ be such a budget-feasible contract and equilibrium. 
    We first observe that it cannot be that $\{\goodaction, \badaction\}\subseteq S$. 
    This is because if $\{\goodaction, \badaction\}\subseteq S$, then $S$ is not budget-feasible due to submodularity of $f$:
    \begin{align*}
        \someindexcontract_{n+1} \geq {c_{\badaction}}/{f_S(\badaction)} \geq  {c_{\badaction}}/{f(\badaction \mid  \goodaction)} \geq {c_{\badaction}}/{\varepsilon} = (3/2) \cdot B > B.
    \end{align*}
    Note that for any $S \subseteq [n]\cup \{\badaction\}$, it holds that $f(S) \le f_1(S)+f_2(S) \le(n/2+2) \cdot \eps$, so proving  $\goodaction\notin S$ is sufficient.
    Assume towards contradiction that $\goodaction\in S$. Since $\{\goodaction, \badaction\} \subsetneq S$, we must have $S_{n+1} = \{\goodaction\}$, and
    \begin{align*}
        \someindexcontract_{n+1} \ge {c_\goodaction}/{f_S(\goodaction)} \ge {c_{\goodaction}} / {f(\{\goodaction\})} 
        =
        B-(n/2) \cdot \eps^2.
    \end{align*} 
    Observe that $|S \cap [n]| \le n/2$ since incentivizing any agent $i\in [n]$ to exert effort takes at least ${c_i}/{f(\{i\})}=\varepsilon^2$, and therefore we can only incentivize ${n}/{2}$ such agents, as the remaining budget is $B-\someindexcontract_{n+1} \le ({n}/{2}) \cdot \eps^2$.
    Since $|S\cap [n]|\le n/2$, $S_{n+1} = \{\goodaction\}$, and $S \ne \{\goodaction\}\cup \specialAgents$, we have $f(\badaction\mid S_{-(n+1)}) =2\eps$.
    Then, we have
    \[
    \someindexcontract_{n+1} \cdot f(\badaction\mid S_{-(n+1)}) - c_{\badaction} \ge \left(B-(n/2) \cdot \eps^2\right) \cdot (2\eps) - (3/2) \cdot B \cdot \eps  \ge  (B/2) \cdot \eps - n \cdot \varepsilon^3 > (B/2) \cdot \varepsilon - n \cdot \varepsilon^2.
    \]
    Moreover, by budget feasibility we have $\someindexcontract_{n+1} \leq B$, and by the definition of $f$ we have $f(\goodaction \mid S_{-(n+1)}) \leq 1/2$, which gives:
    \[
    \someindexcontract_{n+1} \cdot f(\goodaction\mid S_{-(n+1)}) - c_{\goodaction} \le B \cdot (1/2) - (1/2) \cdot \left(B-(n/2) \cdot \eps^2\right) = (n/4) \cdot \eps^2 < n \cdot \varepsilon^2.
    \]
    By our choice of $\varepsilon$, we have $\varepsilon < 4n/B$, and agent $n+1$ would therefore benefit from deviating to $S_{n+1} = \{\badaction\}$, which gives a contradiction.
\end{proof}

We are now ready to prove \Cref{thm:multi-multi-inapprox}.
\begin{proof}[Proof of \Cref{thm:multi-multi-inapprox}]
    Fix a \goodobj\ objective $\varphi$ and a budget $B \in (0,1)$.
    By Yao's principle, it suffices to prove the statement for a deterministic algorithm against a randomized input.
    We consider a randomized instance $\mathcal{I}^{(A')}$, where $A' \subseteq [n]$ is chosen uniformly at random from all subsets of size $n/2$.
    By \Cref{lem:monotone_and_submodular}, this defines a distribution over monotone and submodular instances.
    Now, consider a polynomial-time deterministic algorithm with access to value and demand oracles on this randomized input.
    
    By \Cref{lem:instance_equilibriums} and the definition of \goodobj\ objectives, the optimal value of $\varphi$ is at least a $(1-B)$-fraction of the profit from the budget-feasible contract $(\noindexcontract, \specialAgents \cup {\goodaction}) \in \setofcontracts$, meaning that:
    \begin{align*}
       \MAX\varphi(B) \geq \varphi(\noindexcontract, \specialAgents \cup {\goodaction}) \geq u_P(\noindexcontract, \specialAgents \cup {\goodaction}) \geq (1-B) \cdot f(\specialAgents \cup {\goodaction}) \geq (1-B) \cdot (1/2).
    \end{align*}
Moreover, by \Cref{lem:instance_equilibriums_upper_bound}, the value of $\varphi$ for a contract and equilibrium $(\noindexcontract, S) \in \setofcontracts$ with $S\ne \{\goodaction\}\cup \specialAgents$ is:
\begin{align*}
    \varphi(\noindexcontract, S) \leq f(S) \leq (n/2+2) \cdot \varepsilon.
\end{align*}
Thus, unless the algorithm outputs the equilibrium $\specialAgents \cup {\goodaction}$, it achieves at best an approximation of $(1-B) \cdot (1/2) / ((n/2 + 2) \cdot \varepsilon) > K(n)$, since by our choice of $\varepsilon$ we have $\varepsilon < (1-B) / (K(n) \cdot (n+4))$.

By \Cref{lem:demand_query_by_values}, a polynomial-time algorithm with access to value and demand queries can be simulated using polynomially many value queries.
Thus, it remains to show that any algorithm making only polynomially many value queries cannot output ${\goodaction} \cup \specialAgents$ with better than exponentially small probability.

We assume, without loss of generality, that the algorithm queries the value of the set $S_{-(n+1)} \cup \{\badaction\}$ (where $S$ is the output equilibrium). This is without loss because any algorithm can be modified to perform one additional value query before terminating, without affecting its polynomial query complexity.
We will upper bound the probability that the algorithm queries $A' \cup \{\badaction\}$, thereby establishing the same upper bound on the probability that it achieves a $K(n)$-approximation.

Let $S_1,\dots,S_\ell$ be the sequence of value queries that the (deterministic) algorithm makes on the instance 
$\mathcal{J} = \langle A, \bigsqcup T_i, f_1 + f_2, c \rangle$.
Unless the algorithm queries $\specialAgents\cup \{\badaction\}$, this instance $\mathcal{J}$ is indistinguishable from $\mathcal{I}^{(A')}$. Thus, the probability that the algorithm queries $\specialAgents\cup \{\badaction\}$ is upper bounded by the probability that $\specialAgents\cup \{\badaction\} \in \{S_1,\dots, S_\ell\}$. 
Therefore, by the union bound, the probability of querying $\specialAgents \cup \{\badaction\}$ is at most ${\ell} / \binom{n}{n/2}$.
Since $\ell$ is polynomial in $n$, this probability is exponentially small in $n$, as needed.
\end{proof}

\section{Constant-Factor Approximation for Gross Substitutes Instances} \label{sec:multi-multi-objective-equivalence}
In this section we establish an up-to-a-constant-factor equivalence between any two \goodobj\ objectives, generalizing the result of \cite{feldman2025budget} to the multi-agent combinatorial-actions setting under gross substitutes $f$. 
This equivalence is cast in the following theorem.
\begin{theorem} [Equivalence of All \goodobj\ Objectives and Budgets] \label{thm:multi_multi_objective_equivalence}
Fix any two \goodobj\ objectives $\varphi, \varphi'$ and any two budget $B,B'\in (0,1]$.
For gross substitutes $f$, there exists a poly-time reduction  from $\MAX\varphi(B)$ to $\MAX\varphi'(B')$ that loses only a constant factor in the approximation. This reduction requires value oracle access to $f$.
\end{theorem}

The following corollary follows directly from combining the above equivalence with the poly-time algorithm of \cite{multimulti} for maximizing profit with budget $B=1$ for submodular (and hence also gross substitutes) $f$. 
\begin{corollary}[Constant-Factor Approximations Under Budget Constraints]\label{cor:gsapprox}
    When $f$ is gross substitutes, for any \goodobj\ objective $\varphi$ and any budget $B \in [0,1]$, there exists a polynomial-time algorithm that achieves $O(1)$-approximation  to $\MAX\varphi(B)$ using value queries.
\end{corollary}

To establish~\Cref{thm:multi_multi_objective_equivalence}, we follow a scheme similar to that of~\cite{feldman2025budget}, which proved an analogous result in the binary-action setting. 
In~\Cref{subsec:BR_monotone}, we identify a key property, \emph{best-response monotonicity}, that enables this equivalence.
This property holds both in the binary-action case for submodular $f$ and in the combinatorial-actions case for gross substitutes $f$ (\Cref{lem:bestresponsemonotonicity}). 
Crucially, for combinatorial actions with submodular $f$, best-response monotonicity fails to hold, an observation we exploit in the negative result of~\Cref{sec:inapprox_submodular}.
We use best-response monotonicity to prove a downsizing lemma (\Cref{lem:multi_payment_scaling}) for the combinatorial-actions setting, analogous to the argument in~\cite{feldman2025budget}.

In \Cref{subsc:best_equiv} we define an auxiliary problem, $\multimlr(B)$. 
We show that for any \goodobj\ objective $\varphi$ and any budget $B$, a constant-factor approximation to $\MAX\varphi(B)$ can be obtained either from the solution to $\multimlr(B)$ or from a contract that incentivizes only a single agent to take some actions.

\subsection{Best-Response Monotonicity for Gross Substitutes Instances}\label{subsec:BR_monotone}

In this section we present the main property of gross substitutes that allows us to apply our techniques. 
Roughly speaking, we show that when $f$ is gross substitutes, incentivizing an agent to take a given subset of actions is always cheapest for the principal when all other agents do nothing. This does not hold for submodular $f$, as exemplified by our construction in \Cref{sec:inapprox_submodular}.

\begin{lemma}[Best-Response Monotonicity]\label{lem:bestresponsemonotonicity}
Consider any instance with a gross substitutes $f$.
Fix a contract $\noindexcontract$, an equilibrium $S \in \nash(\noindexcontract)$, and an agent $i \in \agents$.
Take any subset of actions $S'_{-i} \subseteq S_{-i}$.
Then, there exists $S'_i \subseteq T_i$ such that $S_i \subseteq S'_i$ and $S'_i$ is agent $i$'s best response to $S'_{-i}$, i.e., for every alternative set of actions $\tilde S_i \subseteq T_i$, it holds that $\icontract \cdot  f(S'_i \sqcup S'_{-i}) - c(S'_i) \ge \icontract \cdot  f(\tilde{S}_i \sqcup S'_{-i}) - c(\tilde{S}_i)$.
\end{lemma}

\begin{proof}
Fix agent $i \in \agents$, contract $\noindexcontract$, actions $S \in \nash(\noindexcontract)$, and $S'_{-i} \subseteq S_{-i}$ satisfying the above conditions.
Consider two price vectors $\{p_a\}_{a\in \actions}$  and $\{q_a\}_{a\in \actions}$ defined as:
\[
p_a = \begin{cases}
    -1 & \text{if } a\in S_{-i} \\
    2 & \text{if } a \in \actions_{-i} \setminus S_{-i} \\
    (1/\icontract) \cdot c_a & \text{if } a \in \actions_i,
\end{cases}
\qquad
\qquad
    q_a = \begin{cases}
    -1 & \text{if } a \in S'_{-i} \\
    2 & \text{if } a \in \actions_{-i} \setminus S'_{-i} \\
    (1/\icontract) \cdot c_a & \text{if } a \in \actions_i.
\end{cases}
\]

\newpage

We observe that the following equivalences hold for any action profile $S' = S'_{i} \sqcup S'_{-i}$:
\begin{align*}
     & \text{$S'_i$ is agent $i$'s best response to $S'_{-i}$ given $\alpha_i$} \\
    \Longleftrightarrow \quad & S'_i \in \argmax_{U'_i \subseteq \actions_i} \left\{f(U'_i \cup  S'_{-i}) - \sum_{a\in U'_i} (1/\icontract) \cdot c_a\right\} \\
    \Longleftrightarrow \quad & S'=S'_{i} \sqcup S'_{-i} \in \argmax_{V\subseteq \actions} \left\{f(V) - \sum_{a\in V} q_a\right\}.
\end{align*}

The first equivalence is observed by \cite{multimulti}. We next show the second equivalence.
Let $Q \in \argmax_{V \subseteq T} \{f(V) - \sum_{a\in V} q_a\}$. 
We argue that $Q_{-i} = S'_{-i}$; combined with the definition of $q$, this implies the equivalence. 
On the one hand, since $f(U)\le 1$ for all $U\subseteq T$ and $q_a > 1$ for all $a\in T_{-i}\setminus S'_{-i}$, we have $Q_{-i}\subseteq S'_{-i}$.
On the other hand, since $f$ is weakly monotone and $q_a = -1$ for all $a\in S'_{-i}$, it must be that $S'_{-i} \subseteq Q_{-i}$. Thus, $Q_{-i} = S'_{-i}$. 

A similar equivalence also holds with respect to $S$ and $p$. Namely, $U_i \subseteq T_i$ is agent $i$'s best response to $S_{-i}$ given $\alpha_i$ if and only if $U_i\cup S_{-i} \in \argmax_{V\subseteq T} \{f(V)-\sum_{a\in V} p_a\}$. 
Thus, since $S$ is an equilibrium, it is a demand bundle with respect to $f$ and $p$.
Note that $p_a \leq q_a$ for all $a \in T$ by the assumption that $S_{-i}' \subseteq S_{-i}$.
Thus, by the gross substitutes property of $f$, there exists a set $S'_i$ such that $S_i \subseteq S'_i \subseteq T_i$ and 
$S_i'\cup S_{-i}' \in \argmax_{V\subseteq \actions} \{ f(V) - \sum_{a\in V} q_a \}$, as needed.
\end{proof}

\Cref{lem:bestresponsemonotonicity} immediately yields the following important corollary. Recall that $\noindexsubcontract{i}$ is the contract which offers $\icontract$ to agent $i$ and zero to all other agents.

\begin{corollary} 
\label{lem:individual_incentivization}
    Consider an instance with a gross substitutes $f$.
    Fix a contract $\noindexcontract$, an equilibrium $S \in \nash(\noindexcontract)$, and any agent $i \in \agents$.
    Then there exists an equilibrium $S' \in \nash(\noindexsubcontract{i})$ such that $S_i \subseteq S_i'$ and $S'_j = \emptyset$ for any $j \ne i$.
\end{corollary}

\Cref{lem:individual_incentivization} together with \Cref{lem:doubling} enable us to prove a downsizing lemma akin to \cite[Lemma 3.2]{feldman2025budget} for the combinatorial-actions setting. We defer the details of the algorithm and the proof of the lemma to \Cref{apx:downsizing}.

\begin{restatable}[Downsizing Lemma for Combinatorial Actions]{lemma}{multimultidownsizing} \label{lem:multi_payment_scaling} 
    Let $\multiInstance$ be any multi-agent combinatorial-actions instance with gross substitutes $f$.
    For any integer $M \geq 3$ and any $(\noindexcontract, S)\in \setofcontracts$, there exists $(\noindexcontract', S')\in \setofcontracts$ such that:
    \begin{align*}
      \left( \sum_{i\in \agents} \icontract'\leq \frac{5}{M} \cdot \sum_{i\in \agents} \icontract  \quad \text{ or } \quad \exists i\in \agents {\text{ s.t. }} {\noindexcontract' = \noindexsubcontract{i}} {{\text{ and }} S'\subseteq \actions_i}\right) \quad \text{ and } \quad f(S') \geq \frac{1}{2M-2} \cdot f(S).
    \end{align*} 
    Moreover, such a pair $(\noindexcontract', S')\in \setofcontracts$ can be computed in poly-time with value query access to $f$.
\end{restatable}

\subsection{\goodobj\ Objectives Are Equivalent}\label{subsc:best_equiv}
To establish the up-to-a-constant-factor equivalence between any two \goodobj\ under any two budgets, we define two maximization problems (i) finding the optimal budget-feasible contract which incentivizes a single agent, and (ii) finding an optimal budget-feasible contract and an equilibrium where the payment to each agent is at most a $3/4$-fraction of the budget.

\begin{definition}[$\bestsingle{i}{\varphi}(B)$]\label{def:multimulti_singleAgent} 
    For any given objective $\varphi$ and budget $B\in (0,1]$, the problem of $\bestsingle{i}{\varphi}(B)$ is the problem of finding an optimal single-agent contract for agent $i$:
    \[
    \bestsingle{i}{\varphi}(B) = \max_{(\noindexcontract, S) \in \mathcal{C}(B)} \varphi (\noindexcontract, S) \quad\text{subject to}\quad \noindexcontract = \noindexsubcontract{i} \text{ and } S \subseteq T_i.
    \]
\end{definition}
When clear from context, we also use $\bestsingle{i}{\varphi}(B)$ to denote a pair $(\noindexcontract,S) \in \setofcontracts$ maximizing $\varphi$ subject to $\noindexcontract = \noindexsubcontract{i} \text{ and } S \subseteq T_i$.

\begin{definition}[$\multimlr(B)$]\label{def:multimulti_mlr}
    Let $\multiInstance$ be an instance of the multi-agent combinatorial-actions model. For any $B\in (0,1]$, the $\multimlr(B)$ problem is defined as 
    \[
    \multimlr(B) = \max_{(\noindexcontract, S) \in \mathcal{C}(B)} f(S)  \quad\text{subject to}\quad  \icontract \leq 3B/4 \text{ for all } i \in \agents.
    \]
\end{definition}

The best-response monotonicity of gross substitutes instances (\Cref{lem:individual_incentivization}) is crucial for the proof of the following lemma.
\begin{restatable}
[{Decomposition Lemma}]
{lemma}{multimultigoodobjupper} \label{lemma:goodobj_upper_bound_multi_mutli}
 Fix an instance $\multiInstance$ with gross substitutes $f$, a budget $B\in (0,1]$, and a \goodobj  \
    objective $\varphi$.
    It holds that
\[
\MAX\varphi(B) \le 
    2 \cdot \multimlr(B)
    + \max_{i\in \agents} \bestsingle{i}{\varphi}(B).
\]
\end{restatable}
\begin{proof}
    Let $(\noindexcontract^\star, S^\star)\in \setofcontracts$ be a solution to $\MAX\varphi(B)$. 
    
    If $\icontract^\star\le (3/4) \cdot B$ for all agents $i\in \agents$, then $f(S^\star)\le \multimlr(B)$ and we get 
    \begin{align*} \MAX\varphi(B)&=\varphi(\noindexcontract^\star,S^\star) && \text{(by the choice of $(\noindexcontract^\star, S^\star)$)} \\
    &\le f(S^\star)&& \text{(by \Cref{def:goodobj_multi_multi}\ref{def:sandwichproperty})}\\
    &\le \multimlr(B) && \text{(since $(\noindexcontract^\star, S^\star)\in \setofcontracts$ and by our assumption),}
    \end{align*}
    
    as needed. 
    
    Otherwise, let $i\in \agents$ be the agent such that $\icontract^\star > (3/4) \cdot B$; observe that by budget-feasibility there can be at most one such agent. It follows from \Cref{def:goodobj_multi_multi}\ref{def:decompositionproperty} that 
    \[
    \MAX\varphi(B) =\varphi(\noindexcontract^\star, S^\star) \le f(S^\star_{-i})+\varphi(\noindexcontract^\star|_i, S_i^\star).
    \]
    It remains to show $f(S^\star_{-i})\le 2 \cdot \multimlr(B)$, and $\varphi(\noindexcontract^\star|_i, S_i^\star)\le \bestsingle{i}{\varphi}(B)$. 

    Let us first bound $f(S^\star_{-i})$.
    Note that \begin{equation} \label{ineq:sum_of_others}
    \sum_{j\in \agents\setminus \{i\}} \someindexcontract^\star_j =\sum_{j\in \agents} \someindexcontract^\star_j - \icontract^\star < B-(3/4) \cdot B = (1/4) \cdot B.
    \end{equation}
    By applying the doubling lemma (\Cref{lem:doubling}) to $\noindexsubcontractstar{-i}$ with $\eps = (1/4) \cdot (B/n)$, we obtain a contract $\noindexcontract' = 2 \noindexsubcontractstar{-i} + \epsilonvec$, where $\epsilonvec = (\eps,\dots,\eps)$, that satisfies the following properties: 
\begin{enumerate}[label=(\roman*),noitemsep,nosep]
\item  $\noindexcontract'$ is budget-feasible, since $\sum_{j\in \agents} \someindexcontract'_j = 2\sum_{j\in \agents \setminus \{i\}} \someindexcontract^\star_j+n\eps \le (1/2) \cdot B + (1/4) \cdot B < B$, where the first inequality follows from Inequality~(\ref{ineq:sum_of_others}).
\item The payment to every agent $j \in \agents \setminus \{i\}$ is $\someindexcontract'_j = 2 \cdot \someindexcontract^\star_j + \eps \leq (1/2) \cdot B + (1/4) \cdot B \le (3/4) \cdot B$.
\item The payment to agent $i$ is $\icontract' = \eps < (3/4) \cdot B$.
\item Any equilibrium $S'\in \nash(\noindexcontract')$ satisfies $f(S') \ge (1/2) \cdot f(S^\star_{-i})$, by \Cref{lem:doubling}.
\end{enumerate}
    Thus, taking any $S'\in \nash(\noindexcontract')$, we get 
    \begin{align*}
    f(S^\star_{-i})&\le 2 \cdot f(S') && \text{(by (iv))} \\
    &\le 2 \cdot \multimlr(B) &&\text{(by (i)-(iii))}.
    \end{align*}

    Next, let us bound $\varphi(\noindexcontract^\star|_i, S_i^\star)$.
    Observe that, from \Cref{lem:individual_incentivization} it holds that there exists an equilibrium $S'$ of the contract $\noindexcontract^\star|_i$ such that $S'\subseteq T_i$ and $S_i^\star \subseteq S'$. Thus,
    \[
    \varphi(\noindexcontract^\star|_i, S_i^\star)\le \varphi(\noindexcontract^\star|_i,S') \le \bestsingle{i}{\varphi}(B),
    \]
    where the first inequality is by \Cref{def:goodobj_multi_multi}\ref{def:weakly_increasing_property}. This concludes the proof.
\end{proof}

Our reductions will take the better contract between the one achieved by (approximately) solving the problem we reduce to, $\multimlr(B)$, and the best single-agent contract. In the binary-actions case, the best single-agent contract for agent $i\in \agents$ is simply the ratio between their cost and the success probability when only $i$ exerts effort. In the combinatorial-actions case, solving $\bestsingle{i}{\varphi}(B)$ is not as straightforward, but the following lemma shows we can still do so in polynomial time when $f$ is gross substitutes.

\begin{restatable} {lemma}{gssingleagent} \label{lem:gs_single_agent}
    Fix some objective $\varphi$ and budget $B\in (0,1]$. When $f$ is gross substitutes, there exists a poly-time algorithm which (exactly) solves $\bestsingle{i}{\varphi}(B)$ with value oracle access to $f$.
\end{restatable}
The proof of \Cref{lem:gs_single_agent} relies on the ``critical point'' analysis of \cite{dutting2022combinatorial}, who proved the lemma for the special case of maximizing profit for $B=1$, and is deferred to \Cref{app:gs_single_agent}.
We now have all the building blocks for our reductions. We begin by proving a reduction from $\MAX\varphi(B)$ to $\multimlr(B)$.

\begin{restatable}[Reduction to $\multimlr(B)$]{lemma}{reductiontomrlmultimulti}
\label{lem:reduction_to_mrl_multi_multi}
    Fix an instance $\multiInstance$, with gross substitutes $f$, a budget $B\in (0,1]$, and a \goodobj  \
    objective $\varphi$.
    For any $(\noindexcontract, S)\in \setofcontracts$ that is a $\gamma$-approximation to $\multimlr(B)$, let $(\noindexcontract',S')$ be the result of applying the downsizing lemma (\Cref{lem:multi_payment_scaling}) to $(\noindexcontract,S)$ with $M=6$. Then, it holds that one of $\{\bestsingle{i}{\varphi}(B)\}_{i\in \agents} \cup \{(\noindexcontract',S')\}$ is a $(120\gamma+1)$-approximation to $\MAX\varphi(B)$. 
\end{restatable}
\begin{proof}
    Let $(\noindexcontract,S)\in \setofcontracts$ be a $\gamma$-approximation to $\multimlr(B)$ and let $(\noindexcontract', S')$ be the result of applying \Cref{lem:multi_payment_scaling} to $(\noindexcontract,S)$ with $M=6$. This yields a contract-equilibrium pair $(\noindexcontract', S')\in \setofcontracts$ such that $f(S') \ge (1/10) \cdot  f(S)$ and either $\noindexcontract' = \noindexsubcontract{i}$ for some $i\in \agents$ or  $\sum_{i\in \agents}\icontract' \le (5/6) \cdot \sum_{i\in \agents} \icontract \le (5/6) \cdot B\le 5/6$. 
    In the case where there exists an agent $i \in \agents$ such that $\noindexcontract'=\noindexsubcontract{i}$, we have $\sum_{j\in \agents}\someindexcontract'_j = \icontract \le (3/4) \cdot B\le 3/4$, where the first inequality is because $(\noindexcontract, S)$ satisfies the constraints of $\multimlr(B)$.
    Therefore, it holds in both cases that $\sum_{j\in \agents} \someindexcontract'_j \le 5/6$. Now, it follows that:
    \begin{align*}
    \varphi(\noindexcontract', S') &\ge u_P(\noindexcontract', S') && (\text{since $\varphi$ is a BEST objective}) \\
    &=\left(1-\sum_{i\in \agents} \icontract' \right) f(S') && (\text{by the definition of $u_P$})\\
    &\ge (1/6) \cdot f(S') && (\text{since $\sum_{j\in \agents} \someindexcontract'_j \le 5/6$}) \\
    &\ge (1/60) \cdot f(S) && (\text{since $f(S') \ge (1/10) \cdot  f(S)$}) \\
    &\ge (1/{60  \gamma}) \cdot \multimlr(B) && (\text{since $S$ is $\gamma$-approximation})
    \end{align*}
    Let $V = \max \{\{\bestsingle{i}{\varphi}(B)\}_{i\in \agents} \cup \{\varphi(\noindexcontract',S')\}\}$. 
    By \Cref{lemma:goodobj_upper_bound_multi_mutli}, we get:
    \begin{align*}
    \MAX\varphi(B) &\le 2 \cdot \multimlr(B)
    + \max_{i\in \agents} \bestsingle{i}{\varphi}(B) \\
    &\leq (120\gamma) \cdot \varphi(\noindexcontract',S') + \max_{i\in \agents} \bestsingle{i}{\varphi}(B) \\
    &\leq (120 \gamma + 1) \cdot V.
    \end{align*}
    This concludes the proof. 
\end{proof}

Next, we establish reduction from $\multimlr(B)$ to $\MAX\varphi(B)$.

\begin{restatable}[Reduction from $\multimlr(B)$]{lemma}{reductionfrommrlmultimulti}
\label{lem:reduction_from_mrl_multi_multi}
    Fix an instance $\mathcal{I} = \multiInstance$ with gross substitutes $f$ and two budgets $B, B' \in (0,1]$, and let 
    $\mathcal{I}'=\langle \agents, \{ T_i \}_{i \in A}, f, c'\rangle$, be an instance with scaled costs, $c' = c \cdot (4/3) \cdot (B'/B)$.
    For any \goodobj\ objective $\varphi$, if $(\noindexcontract', S')$ is a $\gamma$-approximation to $\MAX\varphi(B')$ in instance $\mathcal{I}'$, 
    then one of $\{\bestsingle{i}{f}(B)\}_{i\in \agents} \cup \{(\noindexcontract'\cdot (3/4) \cdot (B/B'), S')\}$ is a $50\gamma$-approximation to $\multimlr(B)$ in instance $\mathcal{I}$. 
\end{restatable}

\begin{proof}
    Let $(\noindexcontract^\star, S^\star) \in \setofcontracts$ be a solution to $\multimlr(B)$ in the instance $\mathcal{I}$, and let $(\noindexcontract, S) \in \setofcontracts$ be the result of applying the downsizing lemma (\Cref{lem:multi_payment_scaling}) on $(\noindexcontract^\star, S^\star)$ with $M=14$. 
    By the guarantees of \Cref{lem:multi_payment_scaling}, $f(S) \ge (1/26) \cdot f(S^\star)$ and either
    (i) $\noindexcontract= \noindexsubcontract{i}$ and $S \subseteq \actions_i$ for some $i\in \agents$,
    or (ii) $\sum_{i\in \agents}\icontract \le (5/14) \cdot \sum_{i\in \agents} \icontract^\star \le (5/14) \cdot B$.

    For (i) we get that $\bestsingle{i}{f}(B)$ is a $26$-approximation $\multimlr(B)$, which concludes the proof. 
    
    Suppose that case (ii) holds. 
    Recall that $(\noindexcontract, S)$ is a contract-equilibrium pair in $\mathcal{I}$, and $\mathcal{I}'$ is the same as $\mathcal{I}$ except the costs are scaled by $(4/3) \cdot (B'/B)$. Thus, if we scale the agent payments by the same factor, namely $\noindexcontract^\ddagger = (4/3) \cdot (B'/B) \cdot \noindexcontract$, we have that $(\noindexcontract^\ddagger, S)$ is a contract-equilibrium pair in $\mathcal{I}'$.
    Observe that:
    \begin{align*}
        \sum_{i\in \agents} \icontract^\ddagger \le (4/3) \cdot (B'/B) \cdot (5/14) \cdot B= (10/21) \cdot B'\le \min(10/21, B').
    \end{align*}
    Therefore, $\noindexcontract^\ddagger$ is budget-feasible with respect to $B'$, and the principal's profit from the contract-equilibrium pair $(\noindexcontract^\ddagger, S)$ is at least:
    \begin{align}
        u_P(\noindexcontract^\ddagger,S) \ge \left(1-{10}/{21}\right) \cdot f(S) = (11/21) \cdot f(S) \ge (11/21) \cdot (1/26) \cdot f(S^\star)> (1/50) \cdot f(S^\star). \label{eq:bbb}
    \end{align}
    Let $(\noindexcontract', S')$ be a $\gamma$-approximation to $\MAX\varphi(B')$ in $\mathcal{I}'$. Then it holds that:
\begin{align*}
    f(S') 
    &\ge \varphi_{\mathcal{I}'}(\noindexcontract', S') && (\text{since $\varphi$ is BEST})
    \\
    &\ge \MAX\varphi(B') / \gamma
    && (\text{since $(\noindexcontract',S')$ is a $\gamma$-approximation}) \\
    &\ge  \varphi(\noindexcontract^\ddagger, S) / \gamma
    && (\text{since $\noindexcontract^\ddagger$ is budget-feasible w.r.t. $B'$})
    \\
    &\ge   u_P(\noindexcontract^\ddagger, S)  / \gamma
    && (\text{since $\varphi$ is BEST})
    \\
    &\ge (1/50\gamma)  \cdot f(S^\star)
    && (\text{by \Cref{eq:bbb}})
\end{align*}
    Let $\noindexcontract^\dagger = (3/4) \cdot (B/B') \cdot \noindexcontract'$.
    We show that the pair $(\noindexcontract^\dagger, S')$ makes up a $50\gamma$-approximation to $\multimlr(B)$.
    First, since $(\noindexcontract',S')$ is a contract-equilibrium pair in $\mathcal{I}'$, $(\noindexcontract^\dagger, S')$ make up such a pair with respect to $\mathcal{I}$.
    Since $f(S') \ge (1/50\gamma) \cdot f(S^\star)$, it remains to show that the contract $\noindexcontract^\dagger$ satisfies the feasibility constraints. 
    Observe that $\noindexcontract^\dagger$ is budget-feasible with respect to $B$, as
    $\sum_{i \in \agents} \icontract^\dagger = (3/4) \cdot (B/B') \cdot \sum_{i \in \agents} \icontract' \le (3/4) \cdot B$, where the inequality follows from budget feasibility of $\noindexcontract'$ with respect to $B'$.
    Similarly, 
    for any $i\in \agents$, $\icontract'\le B'$ and thus $\icontract^\dagger =(3/4) \cdot (B/B') \cdot \icontract'\le (3/4) \cdot B$. This concludes the proof.
\end{proof}
Together, the two lemmas above imply \Cref{thm:multi_multi_objective_equivalence}.

\section{FPTAS for Additive Instances} \label{sec:fptas}

In this section, we consider instances with additive $f$. Specifically, we show that the FPTAS for multi-agent binary-action settings of \cite{duetting2022multi}, later generalized to budget constraints by \cite{feldman2025budget}, can be adapted to the combinatorial-actions setting.

\begin{theorem}\label{thm:fptas}
When $f$ is additive, for each objective of maximizing profit, reward, and welfare, there exists a deterministic FPTAS under any budget $B \in [0,1]$, using only value oracle access to $f$.
\end{theorem}

At a high level, our proof for profit maximization relies on discretizing the function $f$ and minimizing payments using the dynamic programming approach as in \cite{duetting2022multi}, while adapting the core ideas to the combinatorial-actions setting. For reward and welfare maximization, \cite{feldman2025budget} noted that these problems can be reduced to the \textsc{Knapsack} problem. However, this is not the case in the combinatorial-actions setting, and so we use the ideas from profit maximization to solve those problems as well.

Fix an additive function $\varphi : 2^T \to [0,1]$ with $\varphi(\emptyset) = 0$ and a real number $b \in [0,1]$.
We will later specify $\varphi(S)$ to be either the reward $f(S)$ or the welfare $f(S)-c(S)$.
We begin by defining a discretization $\tilde{\varphi}$ of $\varphi$. 
Let $\delta = \epsilon/|T|$. Define $\tilde{\varphi}(S) = \sum_{a \in S} \lfloor \varphi(\{a\})/(\delta b)\rfloor (\delta b)$. Note that $\tilde{\varphi}(S)$ is a multiple of $\delta b$ for every $S$.

For each $j \in \{0, \ldots, n\}$ and $x \in \{0, \delta b, 2 \delta b, \ldots, \lceil |T|/\delta \rceil \delta b\}$, we define
\begin{align*}
    A^{(\varphi)}(j,x) = \min_{S, \noindexcontract} \left\{ \sum_{i \in A} \icontract \mid \tilde{\varphi}(S) \geq x, S \in \nash(\noindexcontract), S \subseteq T_1 \cup \ldots \cup T_j \right\}.
\end{align*}
This table can be computed in polynomial time via dynamic programming, as we show below.

\begin{lemma}
    The table $A^{(\varphi)}$ can be computed in polynomial time in $|T|$ and $\epsilon$.
\end{lemma}
\begin{proof}
Observe that for $j = 0$, we have $A(0, 0) = 0$ and $A(0, x) = \infty$ for all $x > 0$.

Let us now fix $j > 0$. Let $T_j = \{a_1, \ldots, a_k\}$ and assume without loss of generality that $c_{a_1} / f(\{a_1\}) \leq c_{a_2} / f(\{a_2\}) \leq \ldots \leq c_{a_k} / f(\{a_k\})$.
Note that for a given payment $\icontract$, agent $i$'s best response $S_i$ for $S_{-i}$ belongs to:
\begin{align*}
    \argmax_{S_i \subseteq T_i} \left\{ \icontract \cdot f(S_i \sqcup S_{-i}) - c(S_i) \right\}
    &= \argmax_{S_i \subseteq T_i} \left\{ \icontract \cdot f(S_i) + \icontract \cdot f(S_{-i}) - c(S_i)  \right\}  \\
    &= \argmax_{S_i \subseteq T_i}\left\{  \icontract \cdot f(S_i) - c(S_i)   \right\} \\
    &= \argmax_{S_i \subseteq T_i} \left\{ \sum_{a \in S_i} (\icontract \cdot f(\{a\}) - c_a ) \right\}
\end{align*}
Therefore, the agent's best response includes all actions such that $c_a \leq \icontract \cdot f(\{a\})$, or equivalently $c_a / f(\{a\}) \leq \icontract$. In particular, any contract incentivizes a prefix of $\{a_1, \ldots, a_k\}$, and hence:
\begin{align*}
    A^{(\varphi)}(j,x) = \min_{\ell \in \{0, 1, \ldots, k\}} \left\{ A\left(j-1, x-\tilde{\varphi}(\{a_1, \ldots, a_{\ell}\})\right) + \frac{c_{a_\ell}}{f(\{a_{\ell}\})} \right\},
\end{align*}
where we treat the term $c_{a_{\ell}}/f(\{a_{\ell}\})$ as $0$ when $\ell = 0$.
This completes the proof.
\end{proof}

We are now ready to prove the main theorem. Below we present the proof for profit maximization, and we defer the proofs for reward and welfare maximization to \Cref{sec:apxfptas}.

\begin{proof}[Proof of \Cref{thm:fptas} for Profit Maximization]
We prove the existence of an FPTAS for profit maximization under budget $B$.
First, let $(\noindexcontract^\star, S^\star)$ be the optimal contract-equilibrium pair under budget $B$. Let $b = \max_{a \in S^\star} f(\{a\})$. Note that the algorithm does not have access to $b$, as it does not have access to the optimal solution, but there are polynomially many candidate values, and we can iterate over all of them. From now on, we assume that we know the value of $b$. 

Note that by definition $A^{(f)}(n,x)$ is increasing in $x$. 
Let:
\begin{align*}
    \overline{x} = \max_{x \in \{ 0, \delta b, 2 \delta b, \ldots, \lceil |T|/\delta \rceil \delta b\}} \left\{x \mid A^{(f)}(n,x) \leq B \right\} \quad\text{and}\quad x' \in \argmax_{x \in \{0, \delta b, \ldots, \overline{x}\}} (1-A^{(f)}(n,x)) \cdot x.
\end{align*}

Let $(\noindexcontract', S')$ be the contract-equilibrium pair (with respect to the original $f$) that minimizes the sum of payments in the definition of $A^{(f)}(n,x')$. We will argue that $(\noindexcontract', S')$ yields a $(1-\epsilon)$-approximation to the optimal profit, which will imply the theorem.

We begin by observing that
$\tilde{f}(S^\star) \leq \lceil |T|/\delta \rceil \delta b$ and that $\tilde{f}(S^\star) \leq \overline{x}$. For the first inequality, note that $\tilde{f}(S^\star) = \sum_{a \in S^\star} \tilde{f}(\{a\}) \leq |S^\star| \cdot b \leq \lceil |T|/\delta \rceil \delta b$. For the second inequality, we have $\sum_{i \in A} \icontract^\star \leq B$ by the budget-feasibility of the optimal contract. Therefore, by the choice of $x'$, we have:
\begin{align}
    (1 - A^{(f)}(n,\tilde{f}(S^\star))) \cdot \tilde{f}(S^\star) \leq (1 - A^{(f)}(n,x')) \cdot x'. \label{eq:choiceofx}
\end{align}

Moreover, we observe that:
\begin{align}
    \tilde{f}(S^\star) \geq \sum_{a \in S^\star} f(\{a\}) - |S^\star| \cdot \delta b \geq \sum_{a \in S^\star} f(\{a\}) - \epsilon \cdot \max_{a \in S^\star} f(\{a\}) \geq (1-\epsilon) \cdot f(S^\star). \label{eq:tildef}
\end{align}
It follows that:
\begin{align*}
(1-\epsilon) \cdot u_P(\noindexcontract^\star, S^\star) &= (1-\epsilon) \cdot \left(1-\sum_{i \in A} \icontract^\star \right) \cdot f(S^\star) && (\text{by the definition of $u_P$}) \\
&\leq  \left(1-\sum_{i \in A} \icontract^\star \right) \cdot \tilde{f}(S^\star) && (\text{by \Cref{eq:tildef}})\\
&\leq \left(1 - A(n,\tilde{f}(S^\star))\right) \cdot \tilde{f}(S^\star) && (\text{by the definition of $A$}) \\
&\leq \left(1 - A(n,x')\right) \cdot x' && (\text{by \Cref{eq:choiceofx}}) \\
&\leq \left( 1 - \sum_{i \in A} \icontract'\right) \cdot \tilde{f}(S') && (\text{by definition of $(\noindexcontract', S')$}) \\
&\leq \left( 1 - \sum_{i \in A} \icontract'\right) \cdot f(S') && (\text{since $\tilde{f}(S) \leq f(S)$ for all $S$}) \\
&= u_P(\noindexcontract', S') && (\text{by the definition of $u_P$})
\end{align*}
Therefore, our algorithm returns a $(1-\epsilon)$-approximation, which completes the proof.
\end{proof}

\bibliographystyle{alpha}
\bibliography{refs}

\appendix

\section{FPTAS for Single-Agent Combinatorial-Actions Instances }\label{sec:mulimulti_singleAgentFPTAS}
\newcommand{\logEpsTau}{\ceil{\log_{(1-\eps)}\tau}}

The special case of contracting a single agent with combinatorial actions, i.e., where the instance $\multiInstance$ satisfies $|A|=1$, was introduced by \cite{dutting2022combinatorial}, and is of independent interest.
Note that for a single agent, a budget-feasible contract $\noindexcontract$ is simply a scalar $\alpha \in [0,B]$, the amount transferred to the agent upon the success of the project.
In response to a contract, the agent may pick any set of actions $S \subseteq T$. 
An FPTAS for maximizing the principal's profit for any monotone $f$ was given in \cite{multimulti} for the (implicit) budget, $B=1$. We adapt their algorithm and some of the arguments to accommodate any budget $B \le 1$.

\begin{theorem}\label{thm:single-agent-fptas}
    Let $\multiInstance$ be a single agent instance, i.e. $|A|=1$, with a monotone reward function $f$. Let $B \in [0,1]$ be any budget.
    Given access to a demand oracle, \Cref{alg:SingleAgentFPTAS} is an FPTAS for the principal's profit under budget $B$.
\end{theorem}

First we introduce some notation: let $\multiInstance$ be an instance with $|A|=1$ and $|T|=m$.
For any contract $B$, let $T_B$ be the collection of \emph{budget-feasible actions}, i.e., subsets of actions which can be incentivized using a budget-feasible contract, $T_B = \{S \subseteq T \mid \exists\alpha \le B \; \text{ s.t. } S \in \nash(\alpha)\}$. 
We also denote the maximal welfare achievable under this collection by $SW_B = \max_{S \in T_B} f(S)-c(S)$.

The FPTAS uses a discretization of the space $[0,B]$. However, to ensure that the running time is polynomial in $m$, it is also required to bound the optimal contract away from 1.
To show this for the case where $B=1$, \cite{multimulti} use the result of \cite{dutting2019simple}, which lower bounds the optimal profit as a function of the number of actions, $2^m$, and the optimal social welfare.
By considering only the budget-feasible actions of $T_B$, the bound can be trivially adapted to the any budget $B \le 1$.

\begin{observation}[\cite{dutting2019simple}]
    Let $\multiInstance$ be an instance with $|A|=1$, $|T|=m$, and a monotone reward function $f: 2^T \to \reals_+$. Let $B \in [0,1]$ be any budget.
    Then, there exists a contract $\alpha \le B$ which achieves profit at least $\frac{SW_B}{2^m}$
\end{observation}

Using the above, \cite{multimulti} have established the following lemma, which can also be applied for the case of $B < 1$. The proof is identical to that of \cite{multimulti}.

\begin{lemma}[\cite{multimulti}]\label{lem:single_agent_opt_bounds}
    Let $\multiInstance$ be an instance with $|A|=1$, $|T|=m$, and let $(\alpha^\star,S^\star)$ be the contract and equilibrium which maximize the principal's profit under budget $B \le 1$.
    Denote $j^\star = \argmax_{j \in S^\star} c_j$.
    It holds that,
    $$
    1-\frac{SW_B}{c_{j^\star}+SW_B}
    \le
    \alpha^\star
    \le 
    1-\frac{SW_B}{m\cdot 2^m(c_{j^\star}+SW_B)}
    $$
\end{lemma}
To prove the correctness of \Cref{alg:SingleAgentFPTAS}, we use the following observation made in \cite{dutting2022combinatorial}.
\begin{observation}[\cite{dutting2022combinatorial}]\label{obs:f_increase_with_alpha}
Consider a single agent, combinatorial-actions setting with a monotone reward function $f$.
Let $\alpha < \alpha'$ be two contracts such that $S_\alpha \subseteq T$ maximizes the agent's utility for contract $\alpha$, and $S_{\alpha'}\subseteq T$ maximize the utility for $\alpha'$, then $f(S_\alpha) \le f(S_{\alpha'})$.
Moreover, if $S_\alpha \ne S_{\alpha'}$, then $f(S_\alpha) < f(S_{\alpha'})$ and $c(S_\alpha) < c(S_{\alpha'})$.
\end{observation}

We are now ready to prove \Cref{thm:single-agent-fptas}.

\begin{proof}[Proof of \Cref{thm:single-agent-fptas}]
First, as an immediate corollary of \Cref{obs:f_increase_with_alpha}, $SW_B = f(S^\dagger) - c(S^\dagger)$, for $S^\dagger$ computed in line 2 of the algorithm. Note that computing $S^\dagger$ requires a single demand query.

Let $\alpha^\star, S^\star$ be the pair of budget-feasible contract and set of actions which maximizes the principal's profit. 
Let $j^\star = \argmax_{j \in S^\star} c_j$.
By \Cref{lem:single_agent_opt_bounds}, the optimal budget-feasible contract, $\alpha^\star$, satisfies 
$$
1-\frac{SW_B}{c_{j^\star}+SW_B}
\le
\alpha^\star
\le 
\min\left\{ B, 1-\frac{SW_B}{m\cdot 2^m(c_{j^\star}+SW_B)} \right\}
$$
Consider the iteration in which $j=j^\star$ in line 4.
It holds that, 
$1-\alpha_{j,\log_{1/(1-\eps)} m\cdot 2^m} \le 1-\alpha^\star \le \frac{1-\alpha_{j,0}}{1-\eps}$. To see that, observe that for $k=0$, we have
$\frac{1-\alpha_{j,0}}{1-\eps}=\frac{SW_B}{c_j+SW_B} \ge 1-\alpha^\star$.
On the other hand, for $k=\log_{1/(1-\eps)} m\cdot 2^m$, we have
$$
1-\alpha_{j,k} =
\max\left\{1-B, (1-\eps)\frac{SW_B}{m\cdot 2^m(c_j+SW_B)} \right\}
\le
\max\left\{1-B, \frac{SW_B}{m\cdot 2^m(c_j+SW_B)} \right\}
\le 
1-\alpha^\star
$$

As $(1-\alpha_{j,k})(1-\eps) \le 1-\alpha_{j,k+1}$, there must be $\alpha_{j,k}$ such that $1-\alpha_{j,k}  \le 1-\alpha^\star \le \frac{1-\alpha_{j,k}}{1-\eps}$.
For this $\alpha_{j,k}$, it holds that $\alpha_{j,k} \ge \alpha^\star$, and by \Cref{obs:f_increase_with_alpha}, $f(S_{j,k}) \ge f(S^\star)$. We can conclude that
$$
(1-\hat{\alpha})f(\hat{S}) \ge (1-\alpha_{j,k})f(S_{j,k}) \ge (1-\eps)(1-\alpha^\star)f(S_{j,k})
\ge (1-\eps)(1-\alpha^\star)f(S^\star),
$$
which concludes the proof.
\end{proof}

\begin{algorithm}[h]
\caption{FPTAS For Single Agent}\label{alg:SingleAgentFPTAS}
\SetAlgoLined
\KwIn{Instance $\multiInstance$ with $|A|=1$, budget $B\le 1$, $\eps\in (0,1)$}
\KwOut{contract $\alpha$ which achieves a $(1-\eps)$-approximation to the optimal profit}
$\hat{\alpha} \gets 0, \hat{S} \gets \emptyset$\;
$S^\dagger \gets \argmax_{S \subseteq T} B\cdot f(S)-c(S)$;

$SW_B \gets f(S^\dagger) - c(S^\dagger)$;

\For{$j \in T$ with $c_j>0$}{
\For{$k = 0\dots \lceil\log_{1/(1-\eps)}m\cdot 2^m\rceil-1$}{
$\alpha_{j,k}\gets \min\{B, 1-(1-\eps)^{k+1}\cdot \frac{SW_B}{c_j + SW_B}$\}\;
$S_{j,k} \gets \argmax_{S \subseteq T}(\alpha_{j,k} f(S) - \sum_{j \in S}c_j)$\;

\If{$(1-\alpha_{j,k})f(S_{j,k}) \ge (1-\hat{\alpha})f(\hat{S})$}{
$\hat{\alpha} \gets \alpha_{j,k}$, $\hat{S} \gets S_{j,k}$\;
}
}
}
\Return{$\hat{\alpha}$}\;
\end{algorithm}

\section{Downsizing Algorithm for Multi-Agent Combinatorial Actions}\label{apx:downsizing}

In this section we provide a downsizing algorithm for the multi-agent combinatorial-actions setting with gross substitutes $f$. 
\cite{feldman2025budget} established a similar downsizing lemma for the multi-agent binary-actions setting with submodular (and even XOS) $f$.
A downsizing lemma specifies how, given a contract $\noindexcontract$ and an equilibrium $S \in \nash(\noindexcontract)$, one can reduce the agents' payments while preserving a fraction of the reward $f(S)$.

\multimultidownsizing*

Note that while the downsizing lemma of \cite{feldman2025budget} for the binary-action setting applies to a broader class of reward functions, in our combinatorial-actions model it is impossible to downsize any set of actions when $f$ is submodular, as shown by our construction in \Cref{sec:inapprox_submodular}. In particular, the set of agents $A'$ together with agent $n+1$ performing $\goodaction$ cannot be incentivized with lower payments, and the reward from any strict subset of these agents fails to approximate $f(A' \cup \goodaction)$.

Even if we restrict $f$ to be gross substitutes, proving the downsizing lemma for the combinatorial-actions setting introduces significant challenges.
Assume we receive a contract $\noindexcontract$ and a set of actions $\bigsqcup_{i\in \agents} S_i = S \in \nash(\noindexcontract)$.
Our goal is to find a contract $\noindexcontract'$ which only pays a subset of the agents $U \subseteq \agents$ and has an equilibrium $S'$ in which for each $i \in U$, performs a (weak) superset of her previous actions, i.e., $S'_i \supseteq S_i$.
In the binary-actions setting, the 
taking $\noindexcontract' = \noindexsubcontract{U}$ trivially yields the above property.
In the combinatorial actions case, this does no longer hold and we need to use the doubling lemma of \cite{multimulti} (\Cref{lem:doubling}) to overcome this. Also, incentivizing a single agent is not as straight-forward as in the binary-actions case, and we utilize the key property of gross substitutes (\Cref{lem:gs_single_agent}) to do so.

\begin{algorithm}[t]
\caption{Downsizing Algorithm for GS Reward}\label{alg:budget_scaling_multimulti}
\SetAlgoLined
\KwIn{integer $M \geq 3$, $(\noindexcontract,S) \in \setofcontracts$, and value oracle access to $f$}
\KwOut{$(\noindexcontract',S') \in \setofcontracts$ with $f(S') \geq f(S)/(M-1)$ and either $\sum_{i \in \agents} \icontract' \leq (5/M)\cdot \sum_{i \in \agents} \icontract'$ or $S'\subseteq T_{i'}$ for some $i' \in \agents$}
set $p \gets \sum_{i \in \agents} \icontract$\;
set $Z \gets \{i \in \agents \mid \icontract > p/M\}$\; 
\If{$f(S_i) \ge (1/(M-1)) \cdot f(S)$ for some $i \in Z$}{\label{line:singleton_if_multimulti}
let $p_a = c_a/\icontract$ for $a \in T_i$ and $p_a = \infty$ otherwise\;
let $S'_i \supseteq S_i$ be a set in the demand for prices $p$\footnotemark\;
\Return{$(\noindexsubcontract{i}, S'_i)$}\; \label{line:return_singleton_multimulti}
}
let $U \gets A \setminus Z$\;
\For{$r = 1, \ldots, M-|Z|-2$}{\label{line:for_r_multimulti} 
    set $W_r \gets \emptyset$\;
    \While{$U$ is non-empty \textbf{and} $\sum_{j \in W_r} \someindexcontract_k \le (p/M)$}{\label{line:while_u_multimulti}
        choose any agent $i \in U$\;
        $U \gets U \setminus \{i\}$\;
        $W_r \gets W_r \cup \{i\}$\;\label{line:add_agent_multimulti}
    }
    \If{$f(W_r) \ge (1/(M-1)) \cdot f(S)$}{\label{line:second_if_multimulti}   
        let $U\gets W_r$\;\label{line:Wr_is_good_multimulti}
        \textbf{break}\;
    }
}
let $\eps \gets (1/(n\cdot M)) \cdot \sum_{i\in \agents} \icontract$\; \label{line:eps_multimulti}
let $\noindexcontract' \gets 2 \cdot \noindexsubcontract{U} +\epsilonvec$\; \label{line:contract_multimulti}
let $S'$ be the result of a demand query to $f$ with prices $p_a = c_a/\icontract'$ for $a\in \actions_i$ and $i \in \agents$\;
\Return{$(\noindexcontract', S')$}\;\label{line:return_multimulti}
\end{algorithm}
\footnotetext{
By \Cref{lem:individual_incentivization}, there exists $S'_i \supseteq S_i$, which is in demand for prices $p$. To find this set (and not another set in demand), we can use a greedy algorithm to find $ S \subseteq T_i \setminus S_i$, which maximizes $f(S \mid S_i) - \sum_{j \in S} p_j$. 
It is well-known that $f(\cdot \mid S_i)$ is gross substitutes and thus this greedy approach will result in a set $S'_i = S_i \cup S$ satisfying the required conditions \cite{PaesLeme17}.
}

In order to prove \Cref{lem:multi_payment_scaling}, we use a lemma which was proved implicitly in \cite{multimulti}. The lemma roughly states that for submodular $f$ subset-stability (see \Cref{def:subsetstable}) is maintained if we restrict the contract to a subset of the agents.
\begin{lemma}[\cite{multimulti}]\label{lem:subStableDownwards}
    Let $\multiInstance$ be a multi-agent combinatorial-actions contract instance, where $f$ is submodular.
    Let $S=\bigsqcup_{i \in \agents} S_i$ be a subset-stable profile of actions with respect to contract $\noindexcontract$.
    For any subset of agents $G \subseteq A$, 
    it holds that $S_G=S\cap \left(\bigcup_{i\in G} T_i\right)$ is subset stable with respect to the contract $\noindexsubcontract{G}$. 
\end{lemma}
\begin{proof}
For any $i \notin G$, the claim trivially holds. 
For $i \in G$, let $S_i'\subseteq S_i$, and denote $X_i = S_i \setminus S'_i$.
It follows that:
\begin{align*}
c(S_i) - c(S_i') = c(X_i) &\le \icontract \cdot f(X_i \mid S'_i \cup S_{-i}) && (\text{by subset stability of $S$ w.r.t. $\noindexcontract$}) \\
&\le \icontract \cdot f(X_i \mid S'_i \cup S_{G\setminus \{i\}}) && (\text{by submodularity of $f$})    
\end{align*}
which concludes the proof.
\end{proof}

We also use the following observation made in \cite{multimulti}: given a contract $\noindexcontract$ one can set prices such that a demand set also form a Nash equilibrium with respect to $\noindexcontract$. In the case where $f$ is gross substitutes, such a demand set can be computed efficiently with value queries to $f$.

\begin{observation}[Claim B.1, \cite{multimulti}]\label{obs:deamnd_implies_NE}
    Fix an instance $\multiInstance$ and a contract $\noindexcontract$. 
    Let $S \subseteq T$ be a set in the demand with respect to prices $p_a = c_a/\icontract$, i.e., $S \in \argmax_{S' \subseteq T} f(S') - \sum_{a \in S'} p_a$.
    Then $S$ is a Nash equilibrium with respect to $\noindexcontract$, i.e., $S \in \nash(\noindexcontract)$.
\end{observation}

We are ready to prove \Cref{lem:multi_payment_scaling}.

\begin{proof}[Proof of \Cref{lem:multi_payment_scaling}]
We show that \Cref{alg:budget_scaling_multimulti} satisfies the conditions of \Cref{lem:multi_payment_scaling}.
The fact that the returned pair $(\noindexcontract',S')$ is such that $S' \in \nash(\noindexcontract')$ follows from \Cref{obs:deamnd_implies_NE}, as the set $S'$ is a demand set with respect to prices $p_a=c_a/\icontract$.

We now move to prove that the payment and reward guarantees are also met. 
This trivially holds in the case that the algorithm returns a singleton $\{i\}$ in \Cref{line:return_singleton_multimulti}.
If the algorithm returns a set $W_r$ (i.e., it reached \Cref{line:Wr_is_good_multimulti}), then by the if condition, we have $f(W_r) \geq (1/(M-1)) \cdot f(S)$. Let $i$ be the last agent added to $W_r$ in \Cref{line:add_agent_multimulti}. By the while-loop condition, we have $\sum_{j \in W_r \setminus \{i\}} \someindexcontract_j \leq (1/M) \cdot \sum_{j \in S} \someindexcontract_j$. Additionally, since $i \notin Z$, it follows that $\icontract \leq (1/M) \cdot \sum_{j \in S} \someindexcontract_j$.    
We conclude that
    \begin{align*}
        \sum_{j \in W_r} \someindexcontract_j 
        = \sum_{j \in W_r \setminus \{i\}} \someindexcontract_j + \icontract  
        \leq \frac{1}{M} \cdot \sum_{j \in S} \someindexcontract_j + \frac{1}{M} \cdot \sum_{j \in S} \someindexcontract_j 
        = \frac{2}{M} \cdot \sum_{j \in S} \someindexcontract_j
        \le
        = \frac{2}{M} \cdot \sum_{j \in \agents} \someindexcontract_j.
    \end{align*}
Thus, the returned contract, $\noindexcontract' = 2\noindexcontract + \epsilonvec$ (\Cref{line:contract_multimulti}), satisfies,
$$
\sum_{i \in \agents} \icontract' 
= 
2 \cdot \sum_{i \in \agents} \icontract' + n \cdot \eps
\le 
\frac{4}{M} \sum_{i \in \agents} \icontract + \frac{1}{M} \sum_{i \in \agents} \icontract
=
\frac{5}{M} \sum_{i \in \agents} \icontract,
$$
where the inequality follows from the definition of $\eps$ in \Cref{line:eps_multimulti}.
We conclude that the conditions of the lemma are satisfied in this case.

Suppose that the algorithm returns the remaining agents $U$, i.e., it did not reach \Cref{line:Wr_is_good_multimulti}. We have:
\begin{align*}
    f(U) &\geq f(S) - \sum_{i \in Z} f(\{i\}) - \sum_{r=1}^{M-|Z|-2} f(W_r) && (\text{by subadditivity of $f$})\\
    &\geq f(S) - \frac{M-2}{M-1} \cdot f(S) && (\text{by \Cref{line:singleton_if_multimulti} and \Cref{line:second_if_multimulti}}) \\
    &= \frac{1}{M-1} \cdot f(S) 
\end{align*}
Since each element added to $W_1, \ldots, W_{M-|Z|-2}$ comes from $U$ and is simultaneously removed from $U$, these sets are pairwise disjoint. Thus, we have:
\begin{align*}
    \sum_{j \in U} \someindexcontract_j 
    &= \sum_{j \in S} \someindexcontract_j - \sum_{i \in Z} \icontract - \sum_{r=1}^{M-|Z|-2} \sum_{i \in W_r} \icontract \\
    &\leq \sum_{j \in S} \someindexcontract_j - (|Z| + (M-|Z|-2)) \cdot (1/M) \cdot \sum_{j \in S}\someindexcontract_j \\
    &= (2/M) \cdot \sum_{j \in S} \someindexcontract_j,
\end{align*}
where the second inequality follows by the definition of $Z$ and the while-loop condition.
As in the previous case, this is enough to conclude that the returned contract, $\noindexcontract'$, satisfies $\sum_{i \in \agents} \icontract' \le (5/M) \cdot \sum_{i \in \agents} \icontract$.
This means that both of the conditions of the lemma are satisfied if the algorithm executes \Cref{line:return_multimulti}, which concludes the proof.
\end{proof}

\section{Missing Proofs}\label{app:missing_proofs}

\subsection{Proof of \Cref{lem:monotone_and_submodular}}\label{app:monotone_and_submodular}

\monotoneAndSubmodular*

\begin{proof}
    Fix a subset $A' \subseteq [n]$ with $|A'| = n/2$.
    Observe that $f_1$ is obviously monotone and submodular, thus it is enough to show that $f':=f_2 - f_3$ is monotone and submodular.
    Additionally, $f_2$ is monotone and submodular, as a uniform $\left({n}/{2}+1\right)$-demand function.

    \textbf{Monotonicity:} Fix $S \subseteq [n] \cup \{\badaction,\goodaction\}$, and $a \notin S$. We will show that $f'(S\cup \{a\}) \ge f'(S)$.
    Consider first the case where $S=\specialAgents \cup \{\badaction\}$.
    Note that $f_3(S) = \varepsilon/2$ and $f_3(S \cup \{a\}) = 0$.
    Therefore, we have $f'(S\cup \{a\})=f_2(S\cup \{a\}) \ge f_2(S) > f'(S)$. 
    
    Next, consider the case where $S\cup \{a\}=\specialAgents \cup \{\badaction\}$.
    Note that $f_3(S) = 0$ and $f_3(S \cup \{a\}) = \varepsilon/2$.
    Thus:
    \begin{align*}
        f'(S\cup \{a\}) 
        = f_2(S\cup \{a\})-{\eps}/{2} 
        =f_2(\specialAgents \cup \{\badaction\}) - \eps/2
        =\eps \cdot \left({n}/{2}+1\right) - \eps/2 
        \ge \eps \cdot ({n}/{2}) = f_2(S) = f'(S).
    \end{align*}
    Finally, if $\specialAgents \cup \{\badaction\} \notin \{S,S\cup \{a\}\}$, then $f_3(S\cup\{a\})=f_3(S)=0$, and  monotonicity is implied by the monotonicity of $f_2$.
    
    \textbf{Submodularity:}
    Let $S\subseteq [n]\cup \{\badaction,\goodaction\}$ and let $a,b\in [n]\cup \{\badaction, \goodaction\}$ such that $a,b \notin S$. We will show $f'(a\mid S) \ge f'(a\mid S\cup \{b\})$, or equivalently
    \[
    f'(\{a\}\cup S) - f'(S) \ge f'(\{a,b\}\cup S) - f'(\{b\}\cup S).
    \]
    Observe that unless one of $S, \{a\}\cup S, \{a,b\} \cup S$ equals $\specialAgents \cup \{\badaction\}$, this inequality is implied by submodularity of $f_2$, as in this case $f_3$ always evaluates to $0$.  
    
    If $S = \specialAgents \cup \{\badaction\}$, the inequality follows from submodularity of $f_2$:
    $$
    f'(a \mid S)= f_2(a \mid S) + {\eps}/{2} \ge f_2(a \mid S \cup \{b\})
    $$

    If $\{a\}\cup S=\specialAgents\cup \{\badaction\}$ (and similarly if $\{b\} \cup S = \specialAgents\cup \{\badaction\}$), then
    \begin{align*}
        f'(\{a\}\cup S) - f'(S) 
        = \eps \cdot \left({n}/{2}+1\right) - {\eps}/{2} - \eps \cdot ({n}/{2}) 
        = {\eps}/{2} \ge 0
        = f'(\{a,b\}\cup S) - f'(\{b\}\cup S).
    \end{align*}
    
    Finally, if $\{a,b\}\cup S=\specialAgents\cup \{\badaction\}$, 
    \begin{align*}
        f'(\{a\} \cup S') - f'(S) 
        &=
        f_2(\{a\} \cup S') - f_2(S) \\
        &=
        \eps 
        >
        {\eps}/{2} \\
        &= f_2(\{a,b\} \cup S') - {\eps}/{2} - f_2(\{b\} \cup S) \\
        &=
        f'(\{a,b\} \cup S') - f'(\{b\} \cup S).
    \end{align*}
    This concludes the proof.
\end{proof}

\subsection{Proof of \Cref{lem:gs_single_agent}} \label{app:gs_single_agent}
\gssingleagent*
In order to prove \Cref{lem:gs_single_agent}, we use the following result of \cite{dutting2022combinatorial}.
\begin{proposition}[\cite{dutting2022combinatorial}]\label{prop:gs_multoiaction_defk}
    In a multi-action problem with a single agent and gross substitutes $f$, let $0 \le \alpha_1 < \dots < \alpha_k \le 1$ be the minimal contracts in which the set of actions maximizing the agent's utility changes. Then, $k = O(n^2)$ and $\alpha_1,\dots,\alpha_k$ can be computed in poly-time with value oracle access to $f$.
\end{proposition}

\begin{proof} [Proof of \Cref{lem:gs_single_agent}]
    Fix $i \in \agents$, let $\icontract^\star \in [0,B]$ be the optimal budget-feasible contract with respect to $\varphi$, and let $S^\star \subseteq T_i$ be the set of actions picked by $i$ in response to $\icontract^\star$.
    As $\varphi$ is weakly decreasing in $\icontract^\star$ (\Cref{def:goodobj_multi_multi}\ref{def:weakly_decreasing_property}), we can assume without loss that $\icontract^\star$ is minimal such that $S^\star$ is incentivized. 

    By \Cref{prop:gs_multoiaction_defk}, the set of contracts for which the agent's best-response changes can be computed in poly-time and  there are at most $O(n^2)$ such contracts. Let $\alpha_1,\dots,\alpha_{\ell}$ be those contracts that are also budget-feasible.
    As $f$ is gross substitutes, computing the agent's best response for $\alpha_j$, namely $S_j \in \argmax_{S \subseteq T_i} \{ \alpha_jf(S)-\sum_{l \in S} c_l \}$, can be done efficiently with value queries to $f$.
    Using value oracle access to $\varphi$, one can efficiently find the optimal set of actions among the $k=O(n^2)$ alternatives.
\end{proof}

\subsection{Missing Proofs from \Cref{sec:model}} \label{sec:proofs_BEST}

\begin{observation}\label{obs:common_objs}
    Any monotone subadditive set function of the actions that is sandwiched between profit and expected reward, including welfare and reward, is a \goodobj\ objective. 
    Additionally, whenever $f$ is subadditive, profit is a \goodobj\ objective.
\end{observation}

\begin{proof}
    Fix an instance $\multiInstance$, a contract $\noindexcontract$, a set of actions $S\subseteq\actions$, and an agent $i\in \agents$. 
    We begin by showing that profit satisfies properties $(i)$-$(iv)$ of  \Cref{def:goodobj_multi_multi}.
    Properties $(i)$, $(iii)$, and $(iv)$ follow directly from the definition of profit, so it remains to show that profit satisfies property $(ii)$ whenever $f$ is subadditive.
    Indeed,
    \[
    \begin{split}
    u_P(\noindexcontract, S) &= \left(1-\sum_{j\in \agents} \someindexcontract_j\right)f(S) \le  \left(1-\sum_{j\in \agents} \someindexcontract_j\right)f(S_{-i})+\left(1-\sum_{j\in \agents} \someindexcontract_j\right)f(S_i) \\
     &\le f(S_{-i})+\left(1-\icontract\right)f(S_i) = f(S_{-i}) + u_P(\noindexsubcontract{i}, S_i),
    \end{split}
    \]
    where the first inequality follows from subadditivity of $f$, and the second holds since $\noindexcontract \ge 0$.
    
    Let $h:2^\actions\rightarrow \reals_{\ge 0}$ be a monotone subadditive set function that is sandwiched between profit and the expected reward, and let $\varphi$ be the objective defined by $\varphi(\noindexcontract, S) = h(S)$ for any contract $\noindexcontract$. 
    Clearly, properties $(i)$, $(iii)$ and $(iv)$ of \Cref{def:goodobj_multi_multi} are satisfied. It remains to show property (2). By subadditivity of $h$, 
    \[
    \varphi(\noindexcontract, S) = h(S) \le h(S_{-i})+h(S_i) \le f(S_{-i})+\varphi(\noindexsubcontract{i}, S_i),
    \]
    as needed.
\end{proof}

\begin{observation}\label{obs:convex_comb_objs}
    Let $\varphi^1,\dots,\varphi^k$ be \goodobj\ objectives, and let $\lambda_1,\dots,\lambda_k\in (0,1)$ be such that $\lambda_1 + \ldots + \lambda_k =1$. It holds that the objective $\varphi$ defined as $\varphi(\noindexcontract, S) = \sum_{i=1}^k\lambda_j \varphi^j(\noindexcontract,S)$ is \goodobj. 
\end{observation}
\begin{proof}
    Fix an instance $\multiInstance$, a contract $\noindexcontract$ and a set of actions $S \subseteq T$. The lower bound of property $(i)$ holds,
    \begin{align*}
        \sum_{j=1}^k\lambda_j \varphi^j(\noindexcontract,S)
        \ge
        \sum_{j=1}^k\lambda_j (1-\sum_{i \in \agents} \icontract)f(S)
        =
        (1-\sum_{i \in \agents} \icontract)f(S),
    \end{align*}
    and also the upper bound
    \begin{align*}
        \sum_{j=1}^k\lambda_j \varphi^j(\noindexcontract,S)
        \le
        \sum_{j=1}^k\lambda_j f(S)
        =
        f(S).
    \end{align*}
    To see that property (ii) holds, fix some agent $i \in \agents$. 
    \begin{align*}
        \sum_{j=1}^k\lambda_j \varphi^j(\noindexcontract,S)
        &\le
        \sum_{j=1}^k\lambda_j (f(S_{-i})+ \varphi^j(\noindexsubcontract{i},S_i))\\
        &=
        f(S_{-i})+\sum_{j=1}^k\lambda_j \varphi^j(\noindexsubcontract{i},S_i)
        =
        f(S_{-i})+ \varphi(\noindexsubcontract{i},S_i),
    \end{align*}
    Clearly, the weak monotonicity in $\noindexcontract$ and $S$ carry over to $\varphi(\noindexcontract,S)$.
    This concludes the proof.
\end{proof}

\subsection{Missing Proofs from \Cref{sec:fptas}}\label{sec:apxfptas}

\begin{proof}[Proof of \Cref{thm:fptas} for Reward and Welfare Maximization]
First, observe that for maximizing reward, 
using the same notation as in the proof for profit maximization, 
it suffices to return the 
contract-equilibrium pair $(\noindexcontract',S')$ that minimizes 
the sum of payments in the definition of $A^{(f)}(n,\overline{x})$, 
instead of $A^{(f)}(n,x')$. This contract is budget-feasible by the choice of $\overline{x}$, and we have
\begin{align*}
    f(S') \geq \tilde{f}(S') \geq \overline{x} \geq \tilde{f}(S^\star) \geq (1-\epsilon) \cdot f(S^\star)
\end{align*}
by the observations made in the proof for profit maximization.

The proof for welfare maximization follows analogously by considering $A^{(f-c)}$ instead of $A^{(f)}$ and taking $b = \max_{a \in S^\star} (f(\{a\}) - c(\{a\}))$.
\end{proof}

\section{Linear Contracts are Without Loss of Generality}\label{sec:linear}

In this section, we demonstrate that when optimizing any \goodobj\ objective in the budgeted multi-agent combinatorial-actions setting, it suffices to consider linear contracts without loss of generality.

\newcommand{\gencontract}{t:\agents \times \{0,1\} \to \reals_{\ge 0}}
\newcommand{\gencontractP}{t':\agents \times \{0,1\} \to \reals_{\ge 0}}

A general contract $\gencontract$ specifies the payment to each agent $i$ in case the project fails, which we denote by $t_i(0)$, and in case it succeeds, denoted by $t_i(1)$.

We begin by defining objectives and \goodobj\ objectives for general contracts, akin to definitions \ref{def:obj_multi_multi} and \ref{def:goodobj_multi_multi}. 
The other definition such as restricted contracts generalize naturally to any $t$.

\begin{definition} [Objectives in the Multi-Agent Combinatorial-Actions Model, General Contracts]
    An \emph{objective} $\varphi$ is defined by a poly-time algorithm that, given a problem instance $\multiInstance$, a contract $\gencontract$, and a subset of actions $S\subseteq T$, outputs a non-negative real number, denoted $\varphi_{\multiInstance}(t,S)$. This algorithm has value oracle access to $f$. 
    We omit the subscript when the instance is clear from context.
\end{definition}

\begin{definition}[Beyond Standard (BEST) Objectives in the Multi-Agent Combinatorial-Actions Model, General Contracts]  \label{def:goodobj_multi_multi_general}
    An objective $\varphi$ belongs to the class of  beyond standard (BEST) objectives if, for any instance $\multiInstance$, it is:
    \begin{enumerate}[label=(\roman*)]
        \item \emph{Sandwiched between profit and reward:} For any $\gencontract$ and $S \subseteq \actions$, it holds that
        $u_P(t,S)\le \varphi(t, S) \le f(S)$.
        \item \emph{Decomposable:} For any $\gencontract$, any $S \subseteq \actions$, and any $i\in \agents$, it holds that $\varphi(t, S) \le f(S_{-i})+\varphi(t|_{i}, S_i)$.
        \item \emph{Weakly increasing in $S$:} For any $\gencontract$ and any $S\subseteq S'\subseteq \actions$, $\varphi(t, S) \le \varphi(t, S')$.
        \item \emph{Weakly decreasing in $t$:} For any $S\subseteq \actions$, any $\gencontract$ and $\gencontractP$ such that $t \le t'$ (coordinate-wise), $\varphi(t, S) \ge \varphi(t', S)$.
    \end{enumerate}
\end{definition}

\begin{claim}\label{cla:linear_wlog}
For any general contract $\gencontract$, and an equilibrium $S \in \nash(t_i)$, the linear contract defined by $\icontract = \max\{0,t_i(1)-t_i(0)\}$ satisfies $S \in \nash(\noindexcontract)$.
\end{claim}
\begin{proof}
Fix a contract $\gencontract$, an agent $i \in \agents$, and a profile of actions taken by all other agents $S_{-i} \subseteq \actions_{-i}$.
If agent $i$'s best-response is $S_i$, then for any $S'_i \subseteq \actions_i$,
$$
t_i(1)f(S_{-i} \sqcup S_i) + t_i(0)(1-f(S_{-i} \sqcup S_i)) - c(S_i) 
\ge
t_i(1)f(S_{-i} \sqcup S'_i) + t_i(0)(1-f(S_{-i} \sqcup S'_i)) - c(S'_i).
$$
Equivalently,
\begin{equation}\label{eq:t1-t0}
(t_i(1)-t_i(0))f(S_{-i} \sqcup S_i) - c(S_i) 
\ge
(t_i(1)-t_i(0))f(S_{-i} \sqcup S'_i) - c(S'_i).
\end{equation}

If $t_i(0) > t_i(1)$, clearly the agent's best-response $S_i$ must satisfy $c(S_i)=0$, otherwise it is dominated by $\emptyset$.
Thus, $S_i$ is also a best-response for the contract $\icontract = 0$, as any set with zero cost.

Otherwise, $S_i$ is a best-response with respect to the contract $\icontract = t_i(1)-t_i(0)$.

We conclude that any profile $S$ which is incentivized by the contract $t$, is also incentivized by the linear contract $\icontract = \max \{t_i(1)-t_i(0),0\}$.
\end{proof}

\begin{corollary}
    Let $\varphi$ be any \goodobj\ objective, for any contract $\gencontract$ and a profile $S \in \nash(t)$, there exists a linear contract $\noindexcontract$ 
    such that $S \in \nash(\noindexcontract)$ and $\varphi(\noindexcontract,S_{\noindexcontract} ) \ge \varphi(t,S)$.
\end{corollary}
\begin{proof}
    Observe that the contract $\icontract = \max \{0, t_i(1)-t_i(0)\}$ satisfies $\noindexcontract \le t$ coordinate-wise. Thus, by property $(iv)$ of \Cref{def:goodobj_multi_multi_general}, $\varphi(\noindexcontract,S) \ge \varphi(t,S)$.
    By \Cref{cla:linear_wlog}, $S \in \nash(\noindexcontract)$.
\end{proof}

\end{document}